\documentclass[11pt]{article}

\usepackage[margin=1in]{geometry}
\usepackage{amsmath, amssymb, amsthm}
\usepackage{booktabs}
\usepackage{siunitx}
\usepackage{bm}
\usepackage{graphicx}
\usepackage{float}
\usepackage{subcaption}  
\usepackage{placeins}
\usepackage{setspace}
\usepackage{natbib}  
\usepackage{hyperref}
\usepackage{tcolorbox}
\usepackage{booktabs}
\usepackage{multirow}
\usepackage{placeins}
\usepackage{float}

\hypersetup{
  colorlinks=true,
  linkcolor=blue,
  citecolor=blue,
  urlcolor=blue
}

\newtheorem{theorem}{Theorem}
\newtheorem{corollary}{Corollary}
\newtheorem{lemma}{Lemma}
\newtheorem{proposition}{Proposition}
\theoremstyle{remark}

\title{\textbf{Model-based Survey Weighting Using Hierarchical Gaussian Processes for Structured Populations}}
\author{Seonghun Lee\\
Columbia University}

\date{October 13 2026}

\begin{document}
\maketitle

\addcontentsline{toc}{section}{Abstract}
\tableofcontents
\newpage

\noindent
\textit{Abstract.} Regression-based survey weighting provides a connection between model-based population estimation and design-based inference through equivalent weights. For nonlinear parametric models, these weights can be defined locally through the sensitivity of a poststratified population estimator to the observed outcomes. However, conventional regression models may inadequately capture nonlinear relationships and complex interactions across sparse poststratification cells, motivating a more flexible approach. Building on this local-equivalence perspective, this paper develops a survey-weighting framework based on hierarchical Gaussian process (HGP) regression. The model combines structured group-level effects with a smooth kernel-based component, allowing information to be shared across related poststratification cells while accommodating nonlinear variation. We derive Gaussian posterior prediction formulas and exploit the structured covariance representation to reduce computational costs. We establish a uniform convergence result for the posterior predictive mean under stated regularity conditions and use this result to support the theoretical analysis of locally equivalent weights, which are derived through a first-order expansion of the population poststratified estimator around the fitted latent mean surface. The resulting framework connects Bayesian nonparametric prediction with an interpretable, derivative-based representation of survey weights. Simulations show that HGP achieves lower population-weighted prediction error under nonparametric structure, while logistic regression performs well under correct parametric specification. Additional analyses examine sensitivity to sample size and frame coverage, subgroup prediction, and population-mean agreement. An application to the Future of Families and Child Wellbeing Study illustrates the method. Overall, this work extends local equivalent-weight construction to nonparametric modeling, providing a theoretical and computational foundation for survey weighting in structured, sparse populations.

\doublespacing

\section{Introduction}

Survey weighting is a fundamental tool for producing population-representative estimates in the presence of unequal sampling probabilities, nonresponse, and undercoverage. In our previous work, we developed a model-based framework for constructing survey weights using logistic regression (Lee, \citeyear{Lee2026}). While linear regression admits an exact weighted-mean representation, we constructed locally equivalent weights for logistic regression through a first-order expansion of the nonlinear population estimator. This framework supports the study of effective sample size, asymptotic linearity, and variance estimation. However, logistic regression imposes a linear predictor on the log-odds scale and may be restrictive when relationships involve complex interactions or nonlinear variation across structured populations. The present paper extends this local-equivalence perspective to hierarchical Gaussian process models, enabling locally equivalent-weight construction under more flexible, nonparametric representations of population structure.

In this paper, we extend this line of work by developing a more flexible framework for model-based survey weighting using hierarchical Gaussian processes (HGPs). HGPs accommodate nonlinear relationships while borrowing strength across related groups through hierarchical structure and kernel-based smoothing. This flexibility is particularly valuable in survey settings where small subgroup sizes and sparsely observed poststratification cells can limit the reliability of conventional regression and poststratification methods (Little, \citeyear{Little1993}; Gelman and Little, \citeyear{GelmanLittle1997}; Park et al., \citeyear{ParkGelmanBafumi2004}; Si et al., \citeyear{SiTrangucciGabryGelman2020}). More broadly, our goal is to extend the local-equivalence perspective beyond conventional parametric regression and investigate how flexible nonparametric modeling can support population estimation and survey-weight construction in structured populations.

A key computational feature of our approach is its use of structured covariance matrices that exploit the known grouping of the population. Standard Gaussian process regression can require $\mathcal{O}(n^3)$ computational cost for covariance matrix operations involving $n$ observations (Rasmussen and Williams, \citeyear{RasmussenWilliams2006}). Rather than constructing an unrestricted covariance matrix over all observations, we exploit the hierarchical decomposition and kernel structure associated with the poststratification space (Banerjee et al., \citeyear{BanerjeeGelfandFinleySang2008}; Gelman et al., \citeyear{GelmanCarlinSternDunsonVehtariRubin2013}; Si et al., \citeyear{SiPillaiGelman2015}; M\"akela et al., \citeyear{MakelaSiGelman2018}; Roustant et al., \citeyear{RoustantPadonouDevilleClementPerrinGiorlaWynn2020}). This representation allows posterior inference to leverage the numbers of observations, hierarchical groups, and kernel cells, reducing the computational burden when the structured representation is substantially smaller than the observation-level problem. The resulting framework combines flexible modeling with computationally efficient prediction across large poststratification spaces.

Beyond model construction and computation, a central contribution is the extension of locally equivalent-weight construction to the HGP setting. For nonlinear model-based population estimators, an exact global representation as a weighted sample mean is generally unavailable. Instead, locally equivalent weights can be defined through a first-order expansion of the population estimator around the fitted latent mean surface, linking the estimator's local sensitivity to perturbations in the observed outcomes. Building on this perspective, we derive posterior prediction formulas and establish a uniform convergence result for the posterior predictive mean under stated regularity conditions. These results provide a theoretical foundation for studying the fitted latent surface and its role in constructing locally equivalent weights. The relationship between uniform convergence and the behavior of the resulting weights depends on the additional regularity required for the corresponding linearization.

We evaluate the proposed framework through simulation studies and an application to the Future of Families and Child Wellbeing Study (FFCWS). Our study is situated within the broader literature on regression estimation, calibration, and survey weighting (Little, \citeyear{Little1991}; Deville and Särndal, \citeyear{DevilleSarndal1992}; Gelman, \citeyear{Gelman2007}; Si et al., \citeyear{SiPillaiGelman2015}). Traditional calibration and trimming procedures can help control extreme weights, but their performance may depend on additional modeling and tuning choices, particularly in sparse settings (Little and Vartivarian, \citeyear{LittleVartivarian2003}; Haziza and Beaumont, \citeyear{HazizaBeaumont2017}; Chen et al., \citeyear{ChenLiWu2020}). Similarly, logistic regression can encounter difficulties in finely stratified settings with sparse cells or separation (Albert and Anderson, \citeyear{AlbertAnderson1984}; Chen et al., \citeyear{ChenElliottLittle2010}; Kosmidis and Firth, \citeyear{KosmidisFirth2021}). Our simulation studies therefore compare HGP with unregularized logistic regression and logistic regression with Gaussian-prior regularization under parametric and nonparametric data-generating mechanisms. We examine population-weighted prediction error, sensitivity to sample size and eligible-frame coverage, subgroup- and cell-level prediction accuracy, and agreement between model-based weighted estimates and population benchmarks. The results show that the relative performance of the methods depends on the underlying population structure: logistic regression performs well under correct parametric specification, whereas HGP achieves lower population-weighted prediction error under nonparametric structure. The FFCWS application further illustrates the behavior of the proposed framework in an empirical survey setting.

This paper makes three main contributions. First, it develops a structured HGP framework for model-based survey weighting that accommodates nonlinear variation and borrows information across related poststratification cells. Second, it extends locally equivalent-weight construction to this nonparametric setting and establishes theoretical results for posterior prediction that support the analysis of the resulting weights under appropriate regularity conditions. Third, it evaluates the computational and empirical behavior of the proposed method, identifying settings in which nonparametric flexibility can improve population prediction relative to parametric alternatives. Together, these contributions extend the connection between model-based population estimation and survey weighting to structured nonparametric models.

\section{Hierarchical Gaussian Process Regression}

Our goal is to estimate the population mean of an outcome variable \(y\) across structured subgroups, given a sample that sparsely covers the full poststratification space. Let the population be partitioned into \(H\) hierarchical groups (e.g., broader demographic strata for which population-level information is available) and \(K\) kernel groups (i.e., the full set of poststratification cells spanning all demographic combinations). Let \(K_{\mathrm{obs}}\leq K\) denote the number of distinct poststratification cells represented in the sample. In typical survey settings, only a subset of the full poststratification space may be observed (e.g., 1,100 observed cells out of 4,480 total), limiting the information available for direct cell-level estimation (Lee and Gelman, \citeyear{LeeGelman2024}). Related challenges arise in practical survey weighting: Lee and Gelman (\citeyear{LeeGelman2024}) report that small city-level sample sizes complicated raking to demographic margins, leading them to pool some demographic categories.

To overcome this sparsity, we propose a hierarchical Gaussian process (HGP) model that enables information sharing across related groups (Si et al., \citeyear{SiPillaiGelman2015}; Roustant et al., \citeyear{RoustantPadonouDevilleClementPerrinGiorlaWynn2020}; Li et al., \citeyear{LiJonesBanerjeeEngelhardt2025}). Specifically, we assume that the latent outcome-generating process is composed of additive components: a structured group-level effect $\mu_h$ and a smooth kernel-based effect $\eta_k$ modeled by a Gaussian process.

Formally, for each observation \(i = 1, \ldots, n\) in the sample, we define the latent function as
\[
f_i = \mu_{h[i]} + \eta_{k[i]},
\quad \text{for } i = 1, \ldots, n.
\]
The conditional mean is specified through a monotonic link function \(g\):
\[
\mathbb{E}[Y_i \mid X_i] = g^{-1}(f_i).
\]
Here, \(Y \in \mathbb{R}^n\), \(X \in \mathbb{R}^{n \times p}\), and \(g(\cdot)\) is a monotonic link function. The term \(\mu_{h[i]}\) denotes the hierarchical effect for group \(h[i] \in \{1,\ldots,H\}\), while \(\eta_{k[i]}\) denotes the kernel-based effect for cell \(k[i] \in \{1,\ldots,K\}\). Although only \(K_{\mathrm{obs}}\) distinct kernel groups are represented in the sample, predictions for the full set of \(K\) poststratification cells are obtained through the kernel covariance structure, including for cells not observed in the sample.

For continuous outcomes under the identity link, we additionally assume
\[
    Y_i = f_i + \varepsilon_i,
    \qquad \varepsilon_i \sim \mathcal{N}(0, \sigma^2).
\]
For binary outcomes using a non-identity link (e.g., probit or logit), we instead assume
\[
    Y_i \mid f_i
    \sim \operatorname{Bernoulli}\bigl(g^{-1}(f_i)\bigr).
\]
Thus, for binary outcomes,
\[
    g(\mathbb{E}[Y_i \mid X_i]) = f_i,
\]
but the observed binary response remains random conditional on the latent function. The Gaussian observation-noise variance \(\sigma^2\) is used for continuous outcomes and is not an additional observation-noise term in the binary likelihood.

We assume that the hierarchical effects \(\mu\) and the kernel component \(\eta\) are independent, with the following priors on the latent components:
\[
\begin{aligned}
\mu &\sim
\mathcal{N}(\bm{\mu}_{\text{prior}}, \Sigma_\mu),
& \eta &\sim \mathcal{GP}(0, K_\eta),
\end{aligned}
\]
where \(\bm{\mu}_{\text{prior}} \in \mathbb{R}^H\) encodes known population-level means, and \(K_\eta\) is the Radial Basis Function (RBF) kernel defined over the \(K\) structured groups. We denote the induced prior mean vectors of the latent function at the observed and prediction locations by \(\bm{\mu}_f\) and \(\bm{\mu}_{f_*}\), respectively. These vectors are determined by the corresponding hierarchical group memberships and \(\bm{\mu}_{\text{prior}}\). This independence assumption yields an additive covariance structure in which the hierarchical component captures broad group-level variation, while the Gaussian process component captures dependence across poststratification cells through the kernel covariance. Although the two components may explain overlapping patterns in the observed data, their separate covariance structures provide a parsimonious specification for structured prediction.

The squared exponential kernel is computed via:
\begin{equation*}
    [K_\eta]_{ij}
    = \alpha^2 \exp\left(
        -\frac{\|X_i - X_j\|^2}{2\rho^2}
    \right),
\end{equation*}
where \(K_\eta:\mathbb{R}^P \times \mathbb{R}^P \to \mathbb{R}\), given that \(X_i\) and \(X_j\) are \(P\)-dimensional covariate vectors for groups \(i\) and \(j\), \(\alpha\) is the marginal standard deviation, and \(\rho\) is the length-scale parameter.

\section{Gaussian Processes}

\subsection{Gaussian Conditioning and Posterior Prediction}

\textbf{Notation.} For convenience, we denote
\(g(\mathbb{E}[Y \mid X])\) as \(g(\mu)\). Under the model above, \(g(\mu)=f\) whenever \(g\) is invertible on the relevant range.

\begin{theorem}[]
Assume the hierarchical Gaussian process prior defined above. Let \(\mathbf f\) and \(\mathbf f_*\) denote the latent function values at the observed and prediction locations, respectively. Under the joint Gaussian prior on these latent values, their joint distribution is Gaussian:
\[
\begin{bmatrix}
    \mathbf f \\
    \mathbf f_*
\end{bmatrix}
\sim
\mathcal{N}\left(
\begin{bmatrix}
    \boldsymbol{\mu}_{f} \\
    \boldsymbol{\mu}_{f_*}
\end{bmatrix},
\begin{bmatrix}
    \boldsymbol{\Sigma}_{ff}
    & \boldsymbol{\Sigma}_{f f_*} \\
    \boldsymbol{\Sigma}_{f_* f}
    & \boldsymbol{\Sigma}_{f_* f_*}
\end{bmatrix}
\right),
\]
where the covariance matrices are determined by the hierarchical effects and the kernel \(K_\eta\). Then the conditional distribution of the latent function at the prediction locations, given the latent function values at the observed locations, is
\begingroup
\small
\setlength{\abovedisplayskip}{5pt}
\setlength{\belowdisplayskip}{5pt}
\setlength{\jot}{2pt}
\[
f_* \mid \mathbf f
\sim
\mathcal{N}\!\left(
\boldsymbol{\mu}_{f_*}
+ \boldsymbol{\Sigma}_{f_* f}\boldsymbol{\Sigma}_{ff}^{-1}
(\mathbf f-\boldsymbol{\mu}_f),
\;
\boldsymbol{\Sigma}_{f_* f_*}
- \boldsymbol{\Sigma}_{f_* f}\boldsymbol{\Sigma}_{ff}^{-1}
\boldsymbol{\Sigma}_{f f_*}
\right).
\]
\normalsize
\endgroup
We denote this conditional mean and covariance by
\begingroup
\setlength{\abovedisplayskip}{5pt}
\setlength{\belowdisplayskip}{5pt}
\[
f_* \mid \mathbf f
\sim
\mathcal{N}\!\left(
\mathbb{E}[f_* \mid \mathbf f],
\operatorname{Cov}(f_* \mid \mathbf f)
\right).
\]
\endgroup
For continuous outcomes under the identity link, the observed responses satisfy
\begingroup
\setlength{\abovedisplayskip}{5pt}
\setlength{\belowdisplayskip}{5pt}
\[
\mathbf y = \mathbf f + \boldsymbol{\varepsilon},
\qquad
\boldsymbol{\varepsilon} \sim \mathcal{N}(\mathbf 0,\sigma^2\mathbf I_n).
\]
\endgroup
where the observation noise is independent of the latent function. Define
\begingroup
\setlength{\abovedisplayskip}{5pt}
\setlength{\belowdisplayskip}{5pt}
\[
\boldsymbol{\Sigma}_{yy}
= \boldsymbol{\Sigma}_{ff}+\sigma^2\mathbf I_n.
\]
\endgroup
In this case, the posterior prediction of the latent function given the observed responses is
\begingroup
\small
\setlength{\abovedisplayskip}{5pt}
\setlength{\belowdisplayskip}{5pt}
\[
f_* \mid \mathbf y
\sim
\mathcal{N}\!\left(
\boldsymbol{\mu}_{f_*}
+ \boldsymbol{\Sigma}_{f_*f}\boldsymbol{\Sigma}_{yy}^{-1}
(\mathbf y-\boldsymbol{\mu}_f),
\;
\boldsymbol{\Sigma}_{f_* f_*}
- \boldsymbol{\Sigma}_{f_*f}\boldsymbol{\Sigma}_{yy}^{-1}
\boldsymbol{\Sigma}_{ff_*}
\right).
\]
\normalsize
\endgroup
Because \(\mathbf y=\mathbf f+\boldsymbol{\epsilon}\), where
\(\boldsymbol{\epsilon}\sim\mathcal{N}(\mathbf 0,\sigma^2 I_n)\) is independent of the latent GP, we have
\(\boldsymbol{\Sigma}_{f_*y}=\boldsymbol{\Sigma}_{f_*f}\) and
\(\boldsymbol{\Sigma}_{yf_*}=\boldsymbol{\Sigma}_{ff_*}\). Consequently, the predictive distribution for a future continuous observation is
\begingroup
\setlength{\abovedisplayskip}{5pt}
\setlength{\belowdisplayskip}{5pt}
\[
\mathbf y_* \mid \mathbf y
\sim
\mathcal{N}\!\left(
\mathbb{E}[f_* \mid \mathbf y],
\operatorname{Cov}(f_* \mid \mathbf y)
+\sigma^2\mathbf I_{n_*}
\right).
\]
\endgroup
For binary outcomes, the latent function retains its joint Gaussian prior, so the conditional distribution of the prediction-location latent values given the observed-location latent values is Gaussian. However, conditioning directly on the observed binary responses through the Bernoulli likelihood generally yields a non-Gaussian posterior. Thus, the Gaussian conditioning formula above is exact for conditioning on the latent vector, whereas prediction conditional on the observed binary responses requires accounting for the Bernoulli likelihood.
\end{theorem}

\subsection{Efficient Posterior Inference}

\begin{theorem}[]
Suppose the covariance structure admits a representation whose effective dimension is \(H+K_{\mathrm{obs}}\), where \(K_{\mathrm{obs}}\) is the number of distinct kernel groups represented in the sample. Provided that \(H+K_{\mathrm{obs}}\ll n\) and the required covariance matrices are available, the computational cost of solving linear systems involving the \(n\times n\) covariance matrix \(\Sigma_{\mathbf{yy}}\) can be reduced from \(\mathcal{O}(n^3)\) to
\[
\mathcal{O}\left((H+K_{\mathrm{obs}})^2
(n+H+K_{\mathrm{obs}})\right).
\]
\end{theorem}

\noindent\textbf{Efficient implementation for full-frame poststratification prediction.}
Let \(K_{\mathrm{obs}}\) denote the number of distinct poststratification
cells represented in the sample, and let \(K_{\mathrm{full}}\) denote
the total number of cells in the population poststratification frame.
Rather than factorizing a GP covariance matrix over all
\(K_{\mathrm{full}}\) population cells, our implementation uses the
latent GP values at the \(K_{\mathrm{obs}}\) observed cells to predict
over the full population frame.

Let \(\boldsymbol{\eta}_{\mathrm{obs}}\in
\mathbb R^{K_{\mathrm{obs}}}\) denote the latent GP values at the
observed cells, and let
\(\mathbf K_{\mathrm{obs}}\in
\mathbb R^{K_{\mathrm{obs}}\times K_{\mathrm{obs}}}\)
denote their GP covariance matrix. Define
\(\mathbf K_{\mathrm{cross}}\in
\mathbb R^{K_{\mathrm{full}}\times K_{\mathrm{obs}}}\)
as the cross-covariance matrix between the full population frame and
the observed cells. Under the zero-mean GP prior for the component \(\boldsymbol{\eta}\),
the conditional GP mean over the population frame is
\[
\widehat{\boldsymbol{\eta}}_*
=
\mathbf K_{\mathrm{cross}}
\mathbf K_{\mathrm{obs}}^{-1}
\boldsymbol{\eta}_{\mathrm{obs}}.
\]
In computation, the linear system involving
\(\mathbf K_{\mathrm{obs}}\) is solved using its Cholesky factorization,
rather than by explicitly forming its inverse. For the binary-outcome
model, the resulting plug-in predictions are
\[
\widehat p_k
=
g^{-1}\!\left(
\mu_{h(k)}+\widehat{\eta}_{*,k}
\right),
\qquad k=1,\ldots,K_{\mathrm{full}},
\]
where \(g^{-1}\) is the inverse-logit function and \(\mu_{h(k)}\)
is the hierarchical group effect for cell \(k\). This procedure
provides predictions for the entire population frame while restricting
the GP factorization to the \(K_{\mathrm{obs}}\) observed cells.
These plug-in predictions use the conditional GP mean and do not
include additional conditional GP residual draws at the unobserved
population cells.

\begin{tcolorbox}[
  colback=white!95!gray!10,
  colframe=black,
  title=General Algorithm for Structured GP Posterior Inference
]
\setlength{\abovedisplayskip}{0pt}
\setlength{\belowdisplayskip}{0pt}

\noindent \textbf{Step 1: Construct Covariance Matrices}

Compute \(\Sigma_{yy}, \Sigma_{yf_*}, \Sigma_{f_*y},
\Sigma_{f_*f_*}\) using the estimated covariance matrices.

\noindent\textit{Continuous outcome:$\Sigma_{yy}=Z_\mu\Sigma_\mu Z_\mu^\top
+Z_\eta K_\eta Z_\eta^\top+\sigma^2I_n.$}
\\
\noindent\textit{Binary outcome: $\Sigma_{ff}=Z_\mu\Sigma_\mu Z_\mu^\top
+Z_\eta K_\eta Z_\eta^\top.$}
\\
If needed for numerical stability, a small \(\epsilon>0\)
may be added to the diagonal of the latent covariance matrix.

\noindent \textbf{Step 2: Cholesky Decomposition \& Posterior Mean}

For continuous outcomes, compute \(L:=\operatorname{Cholesky}(\Sigma_{yy})\),
where \(\Sigma_{yy}=LL^\top\), and define
\(\mathbf r_y:=\mathbf y-\boldsymbol\mu_f\).
Compute $V_1:=\operatorname{solve}\bigl(L^\top,
\operatorname{solve}(L,\mathbf r_y)\bigr)
=(LL^\top)^{-1}\mathbf r_y.$
\\
Then,
\[
\mathbb E[\mathbf f_*\mid\mathbf y]
=\boldsymbol\mu_{f_*}+\Sigma_{f_*y}V_1.
\]

\noindent \textbf{Step 3: Posterior Covariance}

Compute \(V_2:=\operatorname{solve}(L,\Sigma_{yf_*})\).
Then,
\[
\operatorname{Cov}(\mathbf f_*\mid\mathbf y)
=\Sigma_{f_*f_*}-V_2^\top V_2.
\]

\noindent \textbf{Step 4: Posterior Sampling}

\textit{Identity link:}
\[
\mathbf y_*\mid\mathbf y\sim
\mathcal N\left(
\mathbb E[\mathbf f_*\mid\mathbf y],
\operatorname{Cov}(\mathbf f_*\mid\mathbf y)+\sigma^2I_{n_*}
\right).
\]
\textit{Binary outcome:} Obtain a draw \(\mathbf f_*\) from the latent
posterior \(p(\mathbf f_*\mid\mathbf y)\) using an appropriate
inference method for the Bernoulli likelihood. Then,
\[
\mathbf p_*:=g^{-1}(\mathbf f_*),\qquad
y_{*i}\mid f_{*i}\sim\operatorname{Bernoulli}(p_{*i}).
\]
\end{tcolorbox}

To guide readers through the main computational steps of posterior prediction and sampling, we present a general algorithm for structured GP posterior inference. The section on efficient implementation for full-frame poststratification prediction describes how these computations can be carried out more efficiently by restricting GP covariance factorization to the observed cells.

\subsection{Connection to kernel ridge regression and crossed random effects model}

\begin{proposition}[]
\label{thm:krr-gp}
Consider our hierarchical Gaussian process model with an
\textbf{identity link} \(g(x)=x\), noise variance \(\sigma^2\),
and suppose the group-level prior covariance
\(\Sigma_\mu\to 0\) and
\(\boldsymbol{\mu}_{\mathrm{prior}}=\mathbf{0}\), so that
\(\mu_h\equiv 0\). Let
\[
\mathbf{K} = K_\eta[k,k']
\]
denote the \(n\times n\) kernel matrix evaluated at the \(n\)
observed inputs, and let \(\mathbf{y}\in\mathbb{R}^n\) be
the observed outcomes. Then the GP posterior mean at a new
input \(x^*\) is
\[
\mathbb{E}[f(x^*)\mid\mathbf{y}]
=
\mathbf{k}(x^*,X)^\top
(\mathbf{K}+\sigma^2 I_n)^{-1}\mathbf{y},
\]
which coincides exactly with the Kernel Ridge Regression (KRR)
solution
\[
\hat f(x^*)=\sum_{i=1}^n\alpha_i k_\eta(x_i,x^*),
\qquad
\boldsymbol{\alpha}
=(\mathbf{K}+\lambda I_n)^{-1}\mathbf{y},
\qquad
\lambda=\sigma^2.
\]
\end{proposition}

\begin{proposition}[]
Let \(X\in\mathbb{R}^{K\times p}\) be the matrix of covariates
for the \(K\) hierarchical groups, and set
\[
K_\eta=\tau^2XX^\top.
\]
Assume the identity link \(g(x)=x\), observation noise
\(\varepsilon_i\sim\mathcal{N}(0,\sigma^2)\), and a group-level
intercept prior
\[
\boldsymbol{\mu}\sim
\mathcal{N}(\boldsymbol{\mu}_{\mathrm{prior}},\Sigma_\mu).
\]
The hierarchical GP model
\[
y_i=\mu_{h[i]}+\eta_{k[i]}+\varepsilon_i,
\qquad
\boldsymbol{\eta}\sim\mathcal{N}(\mathbf{0},K_\eta),
\]
admits the random-slope representation
\[
\eta_k=\mathbf{x}_k^\top\mathbf{b}_1,
\qquad
\mathbf{b}_1\sim\mathcal{N}(\mathbf{0},\tau^2I_p).
\]
Consequently, the outcome model can be written as
\[
y_i=\mu_{h[i]}+\mathbf{x}_{k[i]}^\top\mathbf{b}_1
+\varepsilon_i,
\]
where \(\mu_{h[i]}\) is the group-level random intercept and
\(\mathbf{b}_1\) is the shared random-slope coefficient vector.
In a three-index representation, \(h[i]\) may correspond to
the group indexed by \((i,j)\), while \(k[i]\) may correspond
to the group indexed by \((i,j,k)\). Thus, the intercept and
slope contribution depend on their respective group assignments,
while the coefficient vector \(\mathbf{b}_1\) is shared across
groups.

\end{proposition}

\subsection{Uniform Convergence of Posterior Means}

The following lemma gives a general sufficient condition for the posterior mean to converge uniformly to the true latent function at a prescribed rate. It does not establish posterior contraction for the HGP model itself; rather, it shows how a suitable posterior-tail bound translates into a rate for the posterior mean.

\begin{lemma}[]\label{lem:posterior_mean_rate}
Let $\Pi_n$ be a posterior distribution on functions
$f:\mathcal X\to\mathbb R$, and define
\[
\hat f_n(x):=\mathbb E_{\Pi_n}[f(x)].
\]
Let $f_0$ be the true function, and suppose
$\varepsilon_n\downarrow0$. For a constant $M>0$, define
\[
A_n:=\{\|f-f_0\|_\infty\le M\varepsilon_n\}.
\]
Suppose that
\[
\mathbb E_{\Pi_n}\!\left[
\|f-f_0\|_\infty\mathbf 1_{A_n^c}
\right]
=O_{P_{f_0}}(\varepsilon_n).
\]
Then
\[
\|\hat f_n-f_0\|_\infty
=O_{P_{f_0}}(\varepsilon_n).
\]
\end{lemma}

\begin{corollary}[]\label{cor:cor1}
Suppose the conditions of Lemma~\ref{lem:posterior_mean_rate} hold.
In addition, let $\phi=g^{-1}$ be differentiable on $\mathbb R$ and
suppose there exists a constant $L<\infty$ such that
\[
\sup_{t\in\mathbb R}|\phi'(t)|\le L.
\]
Define
\[
\mu_0(x):=\phi(f_0(x)),
\qquad
\hat\mu_n(x):=\mathbb E_{\Pi_n}[\phi(f(x))].
\]
Then
\[
\|\hat\mu_n-\mu_0\|_\infty
=O_{P_{f_0}}(\varepsilon_n).
\]
\end{corollary}

\subsection{Connection to Survey Weighting}
\label{sec:equivalent-weights}
We now present a theorem that connects the regression function \( g^{-1}(f) \) to population-calibrated survey weights.

\begin{lemma}[]
\label{lem:ginv_remainder}
Let $\phi \equiv g^{-1}:\mathbb{R}\to\mathbb{R}$ be the inverse
link. Assume there exist an interval $I\subset\mathbb{R}$ and a
compact set $I_0\subset I^\circ$ such that
\[
f_0(\mathcal{X})\subset I_0,
\]
$\phi$ is continuously differentiable on $I$, and
\[
\sup_{t\in I}|\phi'(t)|<\infty.
\]
Let $\hat f_n$ satisfy the conclusion of Lemma~\ref{lem:posterior_mean_rate}:
\[
\|\hat f_n-f_0\|_\infty
=O_{P_{f_0}}(\varepsilon_n),
\qquad \varepsilon_n\to 0.
\]
For design points $x_1,\ldots,x_n\in\mathcal{X}$, define
\[
z_i:=f_0(x_i),\qquad
\hat z_i:=\hat f_n(x_i),\qquad
z:=(z_1,\ldots,z_n)^\top,\quad
\hat z:=(\hat z_1,\ldots,\hat z_n)^\top,
\]
and let
\[
A(\hat z):=
\operatorname{diag}\big(
\phi'(\hat z_1),\ldots,\phi'(\hat z_n)
\big).
\]
Then, with probability tending to one, the inverse-link
linearization
\[
\phi(z)-\phi(\hat z)
=A(\hat z)(z-\hat z)+r_g
\]
holds, where $\phi$ is applied componentwise and the remainder
satisfies
\[
\|r_g\|_\infty
=o_{P_{f_0}}\big(\|z-\hat z\|_\infty\big),
\qquad
\|r_g\|_\infty=o_{P_{f_0}}(1).
\]
\end{lemma}

\begin{theorem}[]
\label{thm:thm3}

Let
\[
z \equiv g(\mu)
= g\big(\mathbb E[Y\mid X]\big)\in\mathbb R^n
\]
denote the latent mean-scale vector, and let
\(f_*\in\mathbb R^K\) denote the latent GP values at the
\(K\) poststratification cells \(X_*\).
The posterior predictive mean
\[
m_*(z)\equiv \mathbb E[f_*\mid z]
=
\mu_{f_*}
+\Sigma_{f_*f}\Sigma_{ff}^{-1}(z-\mu_f)
\]
is affine in \(z\). Here, \(\Sigma_{ff}\) denotes the prior covariance matrix of the latent function values at the observed locations, and \(\Sigma_{f_*f}\) denotes the cross-covariance between the prediction and observed locations. The map \(m_*(z)\) is the Gaussian conditional mean of the prediction values given the latent vector \(z=f\), rather than the conditional mean given the observed responses. For binary outcomes, we set \(\sigma^2=0\), so \(\Sigma_{yy}=\Sigma_{ff}\) in the covariance construction. Let \(N_k^{\mathrm{pop}}\) denote the population count in cell
\(k\), and define
\[
N=\sum_{k=1}^K N_k^{\mathrm{pop}}>0,
\qquad
N^{\mathrm{pop}}
=
(N_1^{\mathrm{pop}},\ldots,N_K^{\mathrm{pop}})^\top
\in\mathbb R^K.
\]
Define the model-based poststratified estimand
\[
T(z)
\equiv
\frac{1}{N}(N^{\mathrm{pop}})^\top
g^{-1}\big(m_*(z)\big),
\]
where \(g^{-1}\) is applied componentwise.

Let \(\hat z\in\mathbb R^n\) be a fitted latent-scale vector
satisfying
\[
\|z-\hat z\|_\infty=o_{P_{f_0}}(1).
\]
Write \(\phi\equiv g^{-1}\), and define
\[
A(\hat z)
=
\operatorname{diag}\big(
\phi'(\hat z_1),\ldots,\phi'(\hat z_n)
\big)
\in\mathbb R^{n\times n}
\]
and
\[
D_*
=
\operatorname{diag}\big(
\phi'(m_{*,1}(\hat z)),\ldots,
\phi'(m_{*,K}(\hat z))
\big)
\in\mathbb R^{K\times K}.
\]

Assume the conditions of Lemma~\ref{lem:ginv_remainder}
hold for the unit-level inverse-link expansion, and suppose
\(A(\hat z)\) is nonsingular. Assume also that \(K\) is fixed,
the matrix
\[
B\equiv\Sigma_{f_*f}\Sigma_{ff}^{-1}
\]
has uniformly bounded induced infinity norm, and the relevant
cell-level latent values remain with probability tending to one
in a compact neighborhood on which \(\phi'\) is uniformly
continuous.

Define the local equivalent unit weights by requiring the
first-order directional derivatives of the population estimand
and the weighted comparison quantity
\[
S_w(z)=\frac{1}{n}w^\top\phi(z)
\]
to agree at \(\hat z\) in every direction
\(h\in\mathbb R^n\). Assume the resulting weights satisfy
\[
\frac{\|w\|_1}{n}=O_{P_{f_0}}(1).
\]

Then the local equivalent unit weights are
\[
\boxed{
w=nA(\hat z)^{-1}\nabla_z T(\hat z).
}
\]
In particular,
\[
\boxed{
w=
\frac{n}{N}A(\hat z)^{-1}
\Sigma_{ff}^{-\top}
\Sigma_{ff_*}D_*N^{\mathrm{pop}}.
}
\]
Since \(\Sigma_{ff}\) is symmetric, this is equivalently
\[
\boxed{
w=
\frac{n}{N}A(\hat z)^{-1}
\Sigma_{ff}^{-1}
\Sigma_{ff_*}D_*N^{\mathrm{pop}}.
}
\]

These weights satisfy the local survey-weight identity
\begin{equation}
\label{eq:caseB_identity}
\begin{aligned}
T(z)-T(\hat z)
&=
\frac{1}{n}w^\top
\{\phi(z)-\phi(\hat z)\} + o_{P_{f_0}}(\|z-\hat z\|_\infty).
\end{aligned}
\end{equation}
The weights are uniquely determined by the requirement that
the first-order directional derivatives agree in every direction.
\end{theorem}

\begin{itemize}

\item \textbf{Identity (linear).}
If $g^{-1}(\eta)=\eta$, then $(g^{-1})'(\eta)=1$ for all $\eta$, so
\[
A(\hat z)=I_n,
\qquad
D_*=I_K.
\]
Hence the weights reduce to
\[
w
\;=\;
\frac{n}{N}\,
\Sigma_{ff}^{-1}\Sigma_{ff_*}\,
N^{\mathrm{pop}},
\]
which coincides with the standard linear-model local equivalent unit weights.

\item \textbf{Logit.}
If
\[
g^{-1}(\eta)=\mathrm{expit}(\eta)=\{1+\exp(-\eta)\}^{-1},
\]
define the fitted probabilities
\[
p \equiv \mathrm{expit}(\hat z)\in\mathbb{R}^n,
\qquad
p_* \equiv \mathrm{expit}(\eta_*)\in\mathbb{R}^K.
\]
Since $(g^{-1})'(\eta)=\mathrm{expit}(\eta)\{1-\mathrm{expit}(\eta)\}$, we have
\[
A(\hat z)
=
\mathrm{diag}\!\big(p_1(1-p_1),\ldots,p_n(1-p_n)\big),
\]
\[
D_*
=
\mathrm{diag}\!\big(p_{*,1}(1-p_{*,1}),\ldots,p_{*,K}(1-p_{*,K})\big),
\]
and therefore
\[
w
\;=\;
\frac{n}{N}\,
\mathrm{diag}\!\big(p\odot(1-p)\big)^{-1}
\Sigma_{ff}^{-1}\Sigma_{ff_*}\,
\mathrm{diag}\!\big(p_*\odot(1-p_*)\big)\,
N^{\mathrm{pop}},
\]
where $\odot$ denotes elementwise multiplication.

\item \textbf{Probit.}
If $g^{-1}(\eta)=\Phi(\eta)$, let $\varphi$ denote the standard normal pdf.
Then $(g^{-1})'(\eta)=\varphi(\eta)$, so
\[
A(\hat z)
=
\mathrm{diag}\!\big(\varphi(\hat z_1),\ldots,\varphi(\hat z_n)\big),
\qquad
D_*
=
\mathrm{diag}\!\big(\varphi(\eta_{*,1}),\ldots,\varphi(\eta_{*,K})\big),
\]
and
\[
w
\;=\;
\frac{n}{N}\,
\mathrm{diag}\!\big(\varphi(\hat z)\big)^{-1}
\Sigma_{ff}^{-1}\Sigma_{ff_*}\,
\mathrm{diag}\!\big(\varphi(\eta_*)\big)\,
N^{\mathrm{pop}}.
\]
\end{itemize}

The uniform convergence result in Lemma~\ref{lem:posterior_mean_rate} is related to classical posterior contraction theory for Gaussian process priors (van der Vaart and van Zanten, \citeyear{vdVvdZ2008}), which studies concentration of the posterior distribution around the true latent function under suitable regularity conditions. Whereas classical contraction results concern concentration of the posterior distribution, Lemma~\ref{lem:posterior_mean_rate} establishes a corresponding convergence rate for the posterior mean, conditional on its stated posterior-tail bound. In the present work, this result provides a foundation for analyzing the population-level poststratified estimand and, through its local linearization, deriving the corresponding local equivalent survey weights.

In particular, Lemma~\ref{lem:posterior_mean_rate} establishes that the fitted latent surface converges uniformly to the true latent function in probability. Under the additional smoothness and compactness conditions of Lemma~\ref{lem:ginv_remainder}, this convergence yields a first-order inverse-link expansion, with a remainder that is asymptotically negligible relative to the latent-scale estimation error, uniformly over the observed design points. These results provide the regularity needed to characterize the local behavior of the poststratified estimand. Building on this linearization and its other stated assumptions, Theorem~\ref{thm:thm3} represents the first-order sensitivity of the population estimand through local equivalent survey weights. Thus, Theorem~\ref{thm:thm3} translates uniform convergence of the fitted latent surface and regularity of the inverse link into a local, derivative-based representation of survey weights in structured and sparse poststratification settings.

\section{Results}

\subsection{Simulation setup}

We consider a synthetic population constructed from a full cross-classification of city (\(C=10\)), age (\(A=6\)), ethnicity (\(E=4\)), education (\(D=5\)), and marital status (\(M=2\)), yielding \(K=2{,}400\) poststratification cells in total. The population size is set to \(N_{\text{pop}}=100{,}000\). From this population, we draw a simple random sample of size \(n=2{,}000\).

The hierarchical structure consists of \(H=240\) coarser groups defined by city, age, and ethnicity. Due to sparsity induced by the high-dimensional cell structure, only a subset of the full poststratification space is represented in the sample. Let \(K_{\mathrm{obs}}\leq K\) denote the number of distinct poststratification cells observed in the sample. In our simulation, \(K_{\mathrm{obs}}=657\). Model fitting uses the observed data, while predictions for unobserved poststratification cells are informed by the hierarchical prior structure and kernel covariance. This reflects realistic survey settings in which many poststratification cells are unobserved or weakly observed.

All simulation settings are held fixed across scenarios, with variation arising only
from the data-generating mechanism governing the true cell-level outcome probabilities.

\subsection{Simulation Results}
We compare three approaches: hierarchical Gaussian process weighting (HGP), logistic regression without priors (GLM), and logistic regression estimated by maximum a posteriori inference under group-level Gaussian priors (Logit MAP). Two data-generating mechanisms are considered: a correctly specified parametric logistic model and a nonparametric mechanism with latent Gaussian process structure across poststratification cells.

Table~\ref{tab:simulation-results} reports results under two complementary evaluation regimes. Panel~A presents oracle evaluation, in which cell-level mean squared error (MSE) is weighted by the true population cell counts. Panel~B presents method-induced evaluation, in which prediction errors are weighted using each method's own estimated survey weights.

The broader motivation for model-based weighting is that practical survey-weight construction often involves several challenges, including unit nonresponse, missing data, incomplete population benchmarks, and calibration procedures such as raking. These challenges can make it difficult to construct reliable weights using conventional approaches alone, particularly when the population is divided into many sparsely observed covariate-defined cells. Our approach instead uses a structured regression model to capture shared patterns across groups and cells, with the goal of constructing weights that reflect the modeled relationships in the particular sample and target population. This model-based construction is intended as a complementary approach to traditional weighting procedures, rather than a claim that it automatically resolves every source of nonresponse or missing-data bias.

To compare predictive performance in the simulation, let $p_k$ denote the true outcome probability in poststratification cell $k$ and $\hat p_k$ its estimate. We define the population-weighted oracle MSE as
\[
\mathrm{MSE}_{\text{oracle}}
=
\frac{\sum_k N_k(\hat p_k-p_k)^2}{\sum_k N_k},
\]
and the method-induced MSE as
\[
\mathrm{MSE}_{\text{induced}}
=
\frac{\sum_k \hat w_k(\hat p_k-p_k)^2}{\sum_k \hat w_k},
\]
where \(N_k\) denotes the true population count in cell \(k\), 
\(\hat w_i\) denotes the estimated locally equivalent weight for sampled 
unit \(i\), and \(\hat w_k\) denotes the corresponding cell-level locally 
equivalent weight used to evaluate prediction errors in cell \(k\). 
We also report the unweighted MSE, obtained by averaging 
\((\hat p_k-p_k)^2\) equally across cells:
\[
\operatorname{MSE}_{\mathrm{unweighted}}
=
\frac{1}{K}\sum_{k=1}^{K}(\hat p_k-p_k)^2.
\]
The oracle criterion provides a common benchmark for predictive accuracy under the true population distribution, using population counts available by construction in the simulation. The method-induced criterion instead evaluates prediction errors under the weighting distribution generated by each method. Because each method induces its own distribution of weights, this criterion reflects both prediction errors and how weight is allocated across cells; it is therefore a complementary, method-specific diagnostic rather than a directly comparable measure of population-wide predictive accuracy.

Table~\ref{tab:simulation-results} summarizes predictive performance under both oracle and method-induced evaluation regimes. Under the parametric data-generating mechanism, logistic regression without priors achieves the lowest MSE under both unweighted and population-weighted oracle evaluation, consistent with correct parametric specification. The hierarchical GP performs comparably, while the Logit (MAP) estimator exhibits higher prediction error, reflecting the effects of regularization in this setting. 

\begin{table}[H]
\centering
\caption{Simulation results under parametric and nonparametric data-generating mechanisms. Panel A reports oracle evaluation using the true population cell counts, providing a common benchmark for predictive accuracy under the population distribution. Panel B reports method-induced evaluation, in which squared prediction errors are weighted using each method's own estimated survey weights. Panel B provides a method-specific diagnostic of prediction error under the weighting distribution induced by each method and is not directly comparable to Panel A as a measure of population-wide predictive accuracy.}
\label{tab:simulation-results}

\medskip
\textbf{Panel A: Oracle evaluation (population cell counts $N_k$)}
\medskip

\begin{tabular}{llcc}
\toprule
Scenario & Method & Unweighted MSE & Population-weighted MSE \\
\midrule
\multirow{3}{*}{Parametric}
& HGP            & 0.0036 & 0.0030 \\
& GLM (no prior) & \textbf{0.0022} & \textbf{0.0022} \\
& Logit (MAP)    & 0.0086 & 0.0084 \\
\midrule
\multirow{3}{*}{Nonparametric (GP)}
& HGP            & \textbf{0.043} & \textbf{0.035} \\
& GLM (no prior) & 0.049 & 0.049 \\
& Logit (MAP)    & 0.056 & 0.050 \\
\bottomrule
\end{tabular}

\medskip
\textbf{Panel B: Method-induced evaluation (locally equivalent weights $\hat w_k$)}
\medskip

\begin{tabular}{llc}
\toprule
Scenario & Method & MSE using own estimated weights \\
\midrule
\multirow{3}{*}{Parametric}
& HGP            & 0.0025 \\
& GLM (no prior) & \textbf{0.0021} \\
& Logit (MAP)    & 0.0081 \\
\midrule
\multirow{3}{*}{Nonparametric (GP)}
& HGP            & \textbf{0.028} \\
& GLM (no prior) & 0.044 \\
& Logit (MAP)    & 0.033 \\
\bottomrule
\end{tabular}
\end{table}

Under the nonparametric (GP) data-generating mechanism, the hierarchical GP achieves the lowest MSE under both oracle criteria. In particular, its population-weighted MSE is 0.035, compared with 0.049 for GLM and 0.050 for Logit (MAP). Because all methods are evaluated using the same true population cell counts, the oracle criteria provide a common basis for comparing predictive accuracy under the population distribution. These results suggest that the hierarchical GP better captures the latent structure in this simulation setting, particularly when prediction errors are weighted by population cell frequencies.

Panel~B provides a complementary, method-induced evaluation by weighting prediction errors according to each method's own locally equivalent weights. Unlike the oracle criteria in Panel~A, which evaluate all methods under the same true population distribution, this criterion reflects prediction errors under the weighting scheme induced by each method. Under the parametric mechanism, GLM achieves the lowest method-induced MSE (0.0021), followed closely by HGP (0.0025), whereas Logit (MAP) yields a substantially higher value (0.0081). This pattern suggests that, when the logistic specification is correctly specified, the GLM's weighting scheme performs particularly well under this method-specific criterion. 

Under the nonparametric mechanism, HGP achieves the lowest method-induced MSE (0.028), compared with 0.033 for Logit (MAP) and 0.044 for GLM. The reversal in the relative performance of GLM and HGP across the two mechanisms is consistent with the advantage of HGP when the underlying structure departs from the parametric logistic specification. 

Taken together, these results show that the method-induced evaluation depends on the data-generating mechanism and the weights each method constructs. Because each method is evaluated using its own weights, however, a lower method-induced MSE does not necessarily imply better predictive accuracy across the entire population; Panel~A remains the common benchmark for that comparison.

These method-induced MSE values answer a different question from the oracle criteria in Panel~A. Because each method's prediction errors are evaluated under its own estimated weighting distribution, the resulting values reflect both prediction errors and the distribution of weight assigned across cells. A lower method-induced MSE therefore does not necessarily imply better predictive accuracy under the true population distribution; it may also arise because a method assigns less weight to cells where its predictions are less accurate. Panel~B should consequently be interpreted as a method-specific diagnostic of prediction error under the induced weighting distribution, rather than as an independent measure of population-wide predictive accuracy.

The comparison between the two panels is informative because it distinguishes prediction quality under a common population benchmark from prediction error under the weighting distribution generated by each method. In the nonparametric setting, HGP performs best under both criteria, providing complementary evidence of its performance in this simulation. Nevertheless, the oracle results in Panel~A provide the primary evidence for its advantage in population-weighted predictive accuracy, while Panel~B describes how prediction errors behave under the corresponding locally equivalent weights.

The Logit (MAP) results also illustrate that regularization does not uniformly improve predictive performance. Although the prior structure can stabilize estimation, particularly in sparse settings, the MAP estimator does not attain the lowest MSE under either oracle criterion in these simulations. Its method-induced MSE is also higher than that of HGP under the nonparametric mechanism. These findings suggest that regularization alone does not fully capture the nonparametric structure represented by the data-generating process.

To examine sensitivity to sample size and eligible frame coverage, we
further evaluate population-weighted cell-level prediction error under the
nonparametric GP data-generating mechanism. Figure~\ref{fig:prediction-performance}
presents the median population-weighted MSE across simulation replicates as
sample size and eligible frame size vary. Panel~(a) varies sample size while
holding $K_{\mathrm{frame}}=1800$ fixed; Panel~(b) varies $K_{\mathrm{frame}}$.
HGP attains the lowest median MSE across all settings considered. Its
advantage persists as sample size and frame size vary, although the error
does not decrease monotonically with frame size.

Figure~\ref{fig:prediction-performance} illustrates the predictive advantage of HGP under the nonparametric GP data-generating mechanism. In the left panel, HGP attains the lowest median population-weighted MSE at every sample size considered. Prediction error decreases as the sample size increases for all three methods, with particularly pronounced improvements for GLM and Logit MAP. These results demonstrate the practical importance of sample size for cell-level prediction in the settings examined, while HGP maintains its advantage across the sample sizes considered. Because fitting the HGP model is computationally demanding, the sensitivity analysis considers smaller sample sizes ($n \in \{100,300,500\}$) than the main simulation ($n=2{,}000$), with 10 simulation replicates per setting. This design balances computational cost with an assessment of how predictive performance varies with sample size and eligible frame size.

\begin{figure}[H]
\centering
\includegraphics[width=0.98\textwidth]
{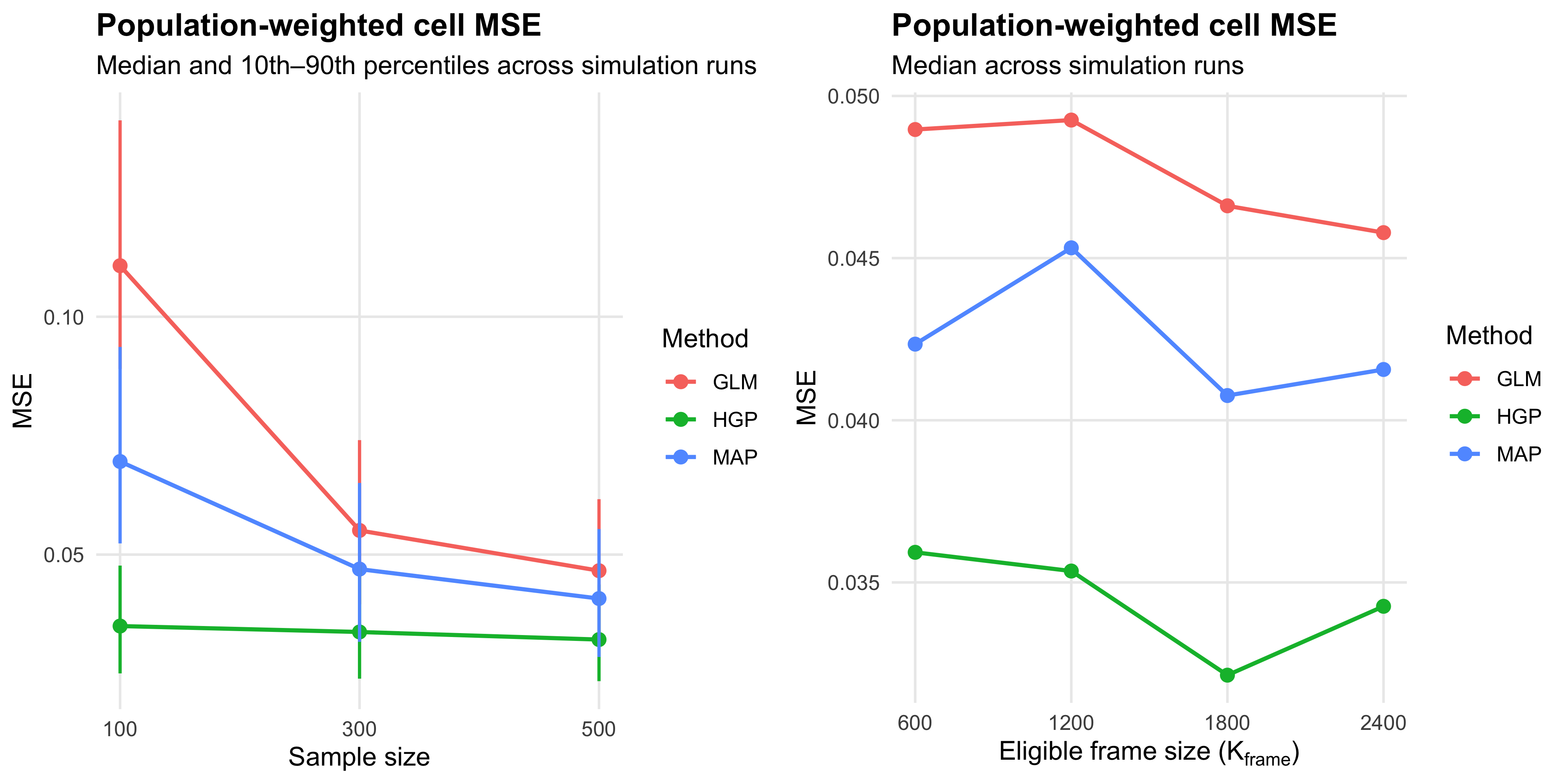}
\caption{Population-weighted prediction error under varying sample sizes
and eligible frame sizes in the nonparametric GP scenario. The left panel
reports the median and 10th--90th percentile intervals across simulation
replicates as sample size varies, with $K_{\mathrm{frame}}=1800$.
The right panel reports the median across replicates as eligible frame
size varies. Both panels use the true population counts $N_k$ to calculate
population-weighted cell-level MSE. Lower values indicate better predictive
accuracy under the population distribution.}
\label{fig:prediction-performance}
\end{figure}

The right panel examines the effect of eligible frame size at fixed sample size at $n=500$. Unlike the consistent improvement observed as sample size increases, enlarging the eligible frame does not produce a monotonic reduction in prediction error for any method. HGP nevertheless maintains the lowest median MSE across all frame sizes considered. This comparison suggests that, in the settings examined, increasing the sample size is more consistently associated with improved predictive accuracy than expanding the eligible frame alone. A larger eligible frame does not necessarily provide additional observed information for estimating cell-level outcome probabilities.


\subsection{Visual diagnostics by $H$-groups under alternative weighting schemes}
\label{subsec:hk-diagnostics}

To complement the population-weighted MSE summaries in Table~\ref{tab:simulation-results},
we examine prediction accuracy at the $H$-group level under two alternative
weighting schemes. The full poststratification partition consists of
$K$-cells defined by city $\times$ age $\times$ ethnicity $\times$
education $\times$ marital status, whereas the coarser $H$-groups are
defined by city $\times$ age $\times$ ethnicity and correspond to the
grouping used for the hierarchical prior on group-specific intercept
offsets. Each $H$-group contains multiple $K$-cells.

We display the 50 $H$-groups with the largest sample sizes, ranked in
descending order of sample size. The first diagnostic uses the true
population cell counts $N_k$ to aggregate both population and predicted
cell rates. The second retains the same population rates as a benchmark
but aggregates predicted cell rates using method-specific estimated
weights $\hat w_k$. Comparing these diagnostics allows us to distinguish
agreement with population-weighted subgroup rates from the local
subgroup rates implied by each method's estimated weighting scheme.
Because the figures focus on the largest-sample $H$-groups, they provide
a descriptive comparison in relatively well-sampled regions rather than
a comprehensive assessment of prediction accuracy across the full
poststratification space.

\paragraph{Population-weighted benchmark.}
Figure~\ref{fig:topH} compares population and predicted rates for the
top 50 $H$-groups, using population cell counts for both series. For
an $H$-group $h$, the population rate and the corresponding predicted
rate are
\[
p_h^{(N)}
=
\frac{\sum_{k\in\mathcal K(h)} N_k p_k}
     {\sum_{k\in\mathcal K(h)} N_k},
\qquad
\hat p_h^{(N)}
=
\frac{\sum_{k\in\mathcal K(h)} N_k \hat p_k}
     {\sum_{k\in\mathcal K(h)} N_k},
\]
where $\mathcal K(h)$ denotes the set of $K$-cells belonging to
$H$-group $h$, $p_k$ is the true population rate, and $\hat p_k$ is
the method-specific predicted rate. Thus, both the red and blue curves
use the same population weights $N_k$, providing a common benchmark
across the three methods.

In the parametric scenario (top row of Figure~\ref{fig:topH}), all
three methods generally reproduce the broad pattern of population
rates across the displayed $H$-groups. Logistic regression without
priors exhibits close agreement in many groups, consistent with its
favorable population-weighted MSE under correct parametric
specification. The hierarchical GP also tracks the population rates
reasonably well, while MAP Logit shows more noticeable discrepancies
in some groups.

In the nonparametric GP scenario (bottom row of Figure~\ref{fig:topH}),
the group-level rates exhibit greater variability, and discrepancies
between population and predicted rates are more apparent. The
hierarchical GP captures several prominent fluctuations, although
some group-level differences remain. Logistic regression misses
some peaks and troughs, while MAP Logit captures certain local
fluctuations but deviates substantially in others. These comparisons
illustrate differences in the recovery of population-weighted subgroup
rates under nonparametric structure, without establishing uniform
dominance by any one method across all groups.

\begin{figure}[H]
\centering
\includegraphics[width=0.98\textwidth]{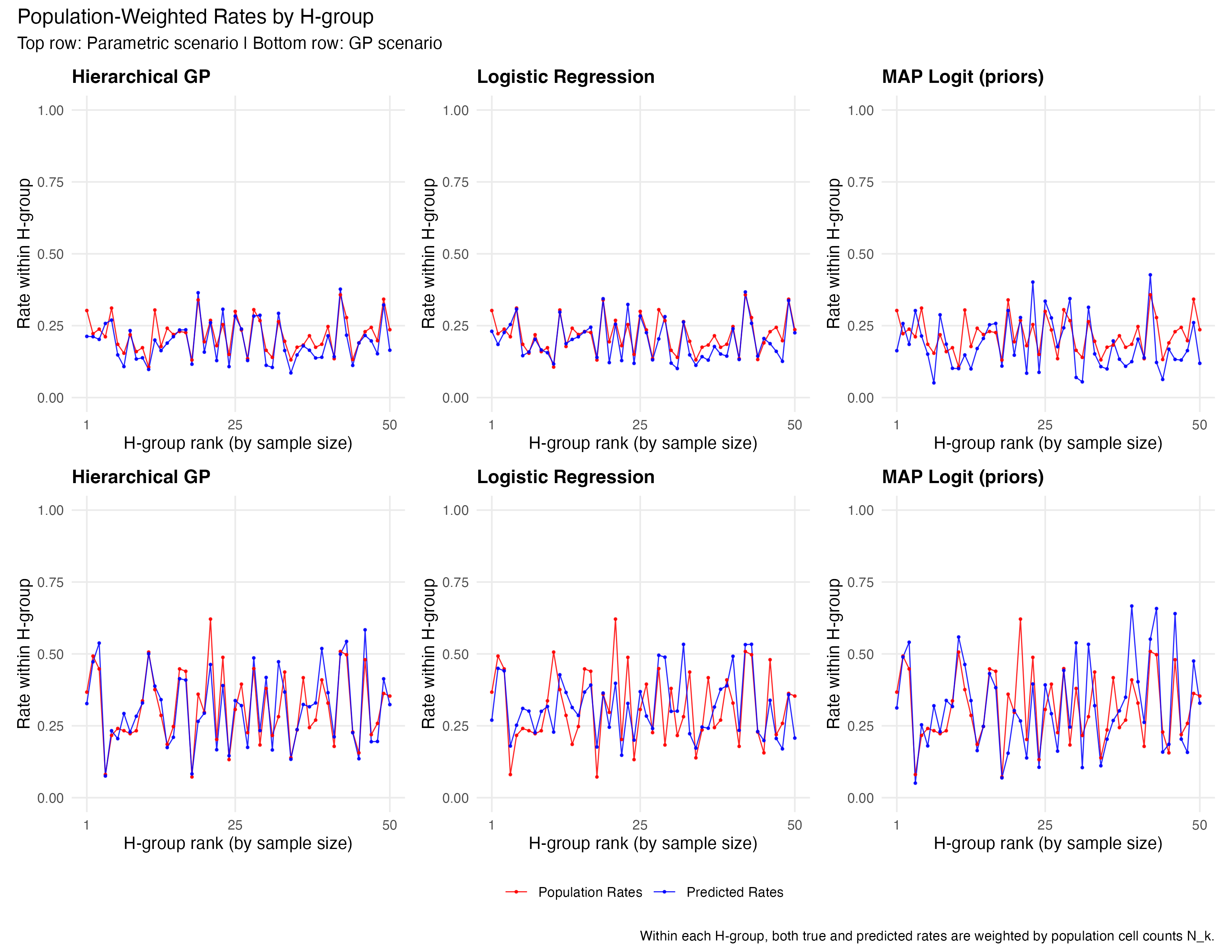}
\caption{Population-weighted benchmark for the top 50 $H$-groups ranked
by sample size. Red denotes the true population rate, and blue denotes
the predicted rate. Both rates are aggregated across constituent
$K$-cells using the true population cell counts $N_k$. The top row
shows the parametric scenario, and the bottom row shows the nonparametric
GP scenario.}
\label{fig:topH}
\end{figure}

Figure~\ref{fig:topH-weighted} uses the same $H$-groups, ordering,
population rates, and true population cell counts as Figure~\ref{fig:topH}.
The difference is that predicted rates are now aggregated using
method-specific estimated cell weights $\hat w_k$ rather than $N_k$.
Specifically, the estimated-weight-adjusted rate for $H$-group $h$ is
\[
\hat p_h^{(w)}
=
\frac{\sum_{k\in\mathcal K(h)} \hat w_k \hat p_k}
     {\sum_{k\in\mathcal K(h)} \hat w_k},
\]
where $\hat w_k$ denotes the estimated cell-level weight associated
with method-specific local equivalent weights. The red curve remains
$p_h^{(N)}$ as defined above; only the aggregation of predicted rates
changes. Consequently, differences between the blue curves in
Figures~\ref{fig:topH} and~\ref{fig:topH-weighted} reflect the change
from population-count weighting to estimated-weight weighting, rather
than a change in the population benchmark or the set of displayed
groups.

\begin{figure}[H]
\centering
\includegraphics[width=0.98\textwidth]{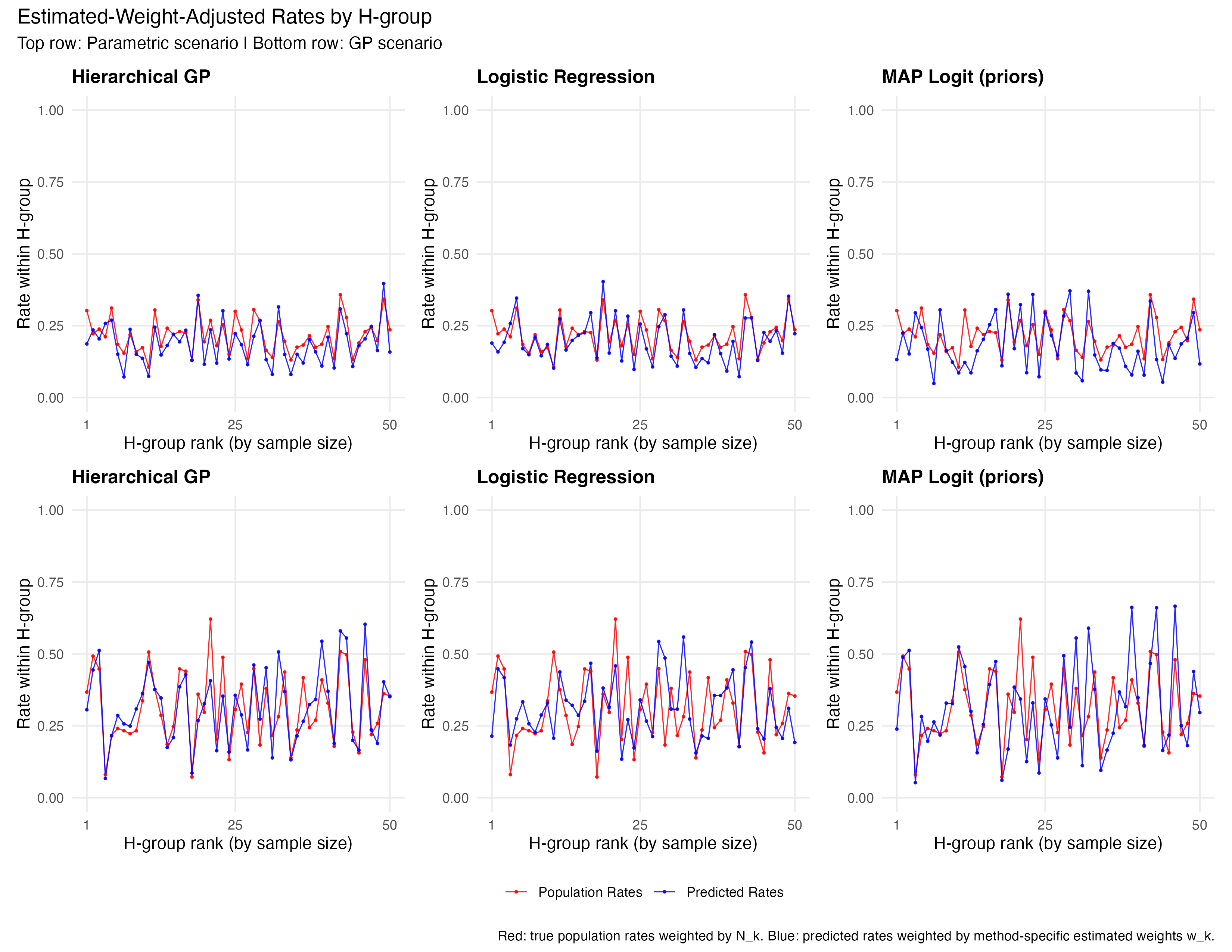}
\caption{Estimated-weight-adjusted predicted rates for the same top
50 $H$-groups ranked by sample size as in Figure~\ref{fig:topH}.
Red denotes the true population rate, aggregated using population
cell counts $N_k$. Blue denotes the predicted rate, aggregated using
method-specific estimated cell weights $\hat w_k$. The top row shows
the parametric scenario, and the bottom row shows the nonparametric
GP scenario. The red curves are identical to those in
Figure~\ref{fig:topH}; differences in the blue curves reflect the
alternative aggregation of predicted cell rates.}
\label{fig:topH-weighted}
\end{figure}

In the parametric scenario (top row of Figure~\ref{fig:topH-weighted}),
logistic regression without priors continues to show close agreement
with the population rates across many groups. The hierarchical GP
also retains a broadly similar pattern under the two weighting
schemes, with relatively small differences in most groups. MAP
Logit, however, exhibits more noticeable discrepancies from the
population benchmark, even though its overall pattern remains
similar under population-count and estimated-weight aggregation.
Thus, changing the aggregation scheme generally preserves the broad
patterns of subgroup rates, but does not eliminate the local
discrepancies between the predicted rates and the population benchmarks.

In the nonparametric GP scenario (bottom row of
Figure~\ref{fig:topH-weighted}), the differences between methods
remain more apparent, reflecting the greater variability in the
underlying population rates. The predicted rates under estimated-weight
aggregation remain broadly similar to those obtained using population
counts, although some local differences are visible. The hierarchical
GP captures several prominent group-level fluctuations, while
logistic regression and MAP Logit exhibit discrepancies in some
groups. Overall, changing the aggregation weights produces relatively
small differences in the displayed subgroup rates, without
substantially changing their broad patterns.

Taken together, Figures~\ref{fig:topH} and~\ref{fig:topH-weighted}
compare two ways of aggregating predicted cell rates within the same
$H$-groups. Figure~\ref{fig:topH} uses population cell counts $N_k$
for both the true and predicted rates, whereas
Figure~\ref{fig:topH-weighted} retains the same population benchmark
but uses method-specific estimated weights $\hat w_k$ for the
predicted rates. The red curves are therefore identical, while the
blue curves differ slightly where the estimated weights change the
relative contributions of constituent cells. The close agreement
between the two sets of blue curves suggests that the resulting
subgroup-level patterns are generally similar under these two
aggregation schemes in the displayed groups. These diagnostics
complement the population-weighted MSE results in
Table~\ref{tab:simulation-results}, but do not establish that the
estimated weights improve prediction or that the same degree of
agreement holds across the full poststratification population.

\subsection{Population-mean agreement and local equivalent weights}
\label{subsec:mean-agreement}

The MSE summaries in Table~\ref{tab:simulation-results} and the
group- and cell-level diagnostics in Figures~\ref{fig:topH} and
\ref{fig:topH-weighted} assess predictive accuracy across the poststratification
space. We complement these comparisons with two diagnostics of
aggregate performance: agreement between the true population mean and
the model-implied weighted mean, and absolute population-mean
prediction error as sample size varies. These quantities address
different aspects of performance and should be distinguished when
evaluating the proposed locally equivalent weights.

Let
\[
\bar p^{\mathrm{pop}}
=
\frac{\sum_k N_k p_k}{\sum_k N_k}
\qquad\text{and}\qquad
\bar p^{w}
=
\frac{\sum_k \hat w_k \hat p_k}{\sum_k \hat w_k},
\]
where $N_k$ denotes the population size of cell $k$, $p_k$ is its
true population rate, $\hat p_k$ is its predicted rate, and $\hat w_k$
is the estimated cell-level weight. The first quantity is the true
population mean, whereas the second is the weighted mean implied by
the estimated weights and predicted cell rates. We summarize their
discrepancy using the signed mean gap
\[
\Delta_{\mathrm{mean}}
=
\bar p^{\mathrm{pop}}-\bar p^w
\]
and the absolute relative gap
\[
\Delta_{\mathrm{rel}}
=
\left|
\frac{\bar p^{\mathrm{pop}}-\bar p^w}
{\bar p^{\mathrm{pop}}}
\right|.
\]
The signed gap indicates the direction of the discrepancy, while the
absolute relative gap expresses its magnitude relative to the true
population mean.

\begin{table}[H]
\centering
\caption{Population-mean agreement under parametric and nonparametric
data-generating mechanisms. The table compares the true population mean
$\bar p^{\mathrm{pop}}$ with the model-implied weighted mean
$\bar p^w$. Relative gaps are reported in absolute value.}
\label{tab:mean-gap}
\begin{tabular}{llccc}
\toprule
Scenario & Method & $\bar p^{\mathrm{pop}}$ & $\bar p^w$ &
Absolute relative gap (\%) \\
\midrule
\multirow{3}{*}{Parametric}
& HGP            & 0.226 & 0.201 & 11.1 \\
& GLM (no prior) & 0.226 & 0.221 & 2.5 \\
& Logit (MAP)    & 0.226 & 0.212 & 6.3 \\
\midrule
\multirow{3}{*}{Nonparametric (GP)}
& HGP            & 0.323 & 0.333 & 3.0 \\
& GLM (no prior) & 0.323 & 0.319 & 1.2 \\
& Logit (MAP)    & 0.323 & 0.314 & 3.0 \\
\bottomrule
\end{tabular}
\end{table}

Table~\ref{tab:mean-gap} shows that aggregate agreement varies across
methods and data-generating mechanisms. Under the parametric mechanism,
logistic regression without priors has the smallest absolute relative
gap (2.5\%), followed by MAP Logit (6.3\%) and HGP (11.1\%). Under the
nonparametric GP mechanism, the relative gaps are smaller for all three
methods: 3.0\% for HGP, 1.2\% for logistic regression without priors,
and 3.0\% for MAP Logit. These results illustrate that improved prediction across a flexible population structure does not necessarily translate into the smallest discrepancy in global mean estimation. Although HGP is designed to accommodate more complex population structure, logistic regression without priors achieves closer agreement for the global mean in both scenarios considered. These results illustrate that agreement for a single aggregate functional need not
coincide with the ranking of methods according to cell-level or
population-weighted prediction error.

We next examine how aggregate predictive accuracy changes with sample
size under the nonparametric GP data-generating mechanism. Figure
\ref{fig:popmean-error} reports the absolute difference between the
true population mean and the population-weighted mean of the predicted
cell rates:
\[
\left|
\widehat{\bar p}^{\,\mathrm{pop}}
-
\bar p^{\mathrm{pop}}
\right|
=
\left|
\frac{\sum_k N_k\hat p_k}{\sum_k N_k}
-
\frac{\sum_k N_k p_k}{\sum_k N_k}
\right|.
\]
Unlike the preceding table, this diagnostic uses the true population
counts $N_k$ to aggregate predicted rates and does not use the
estimated survey weights $\hat w_k$ in the aggregation.

\begin{figure}[H]
\centering
\includegraphics[width=0.7\textwidth]{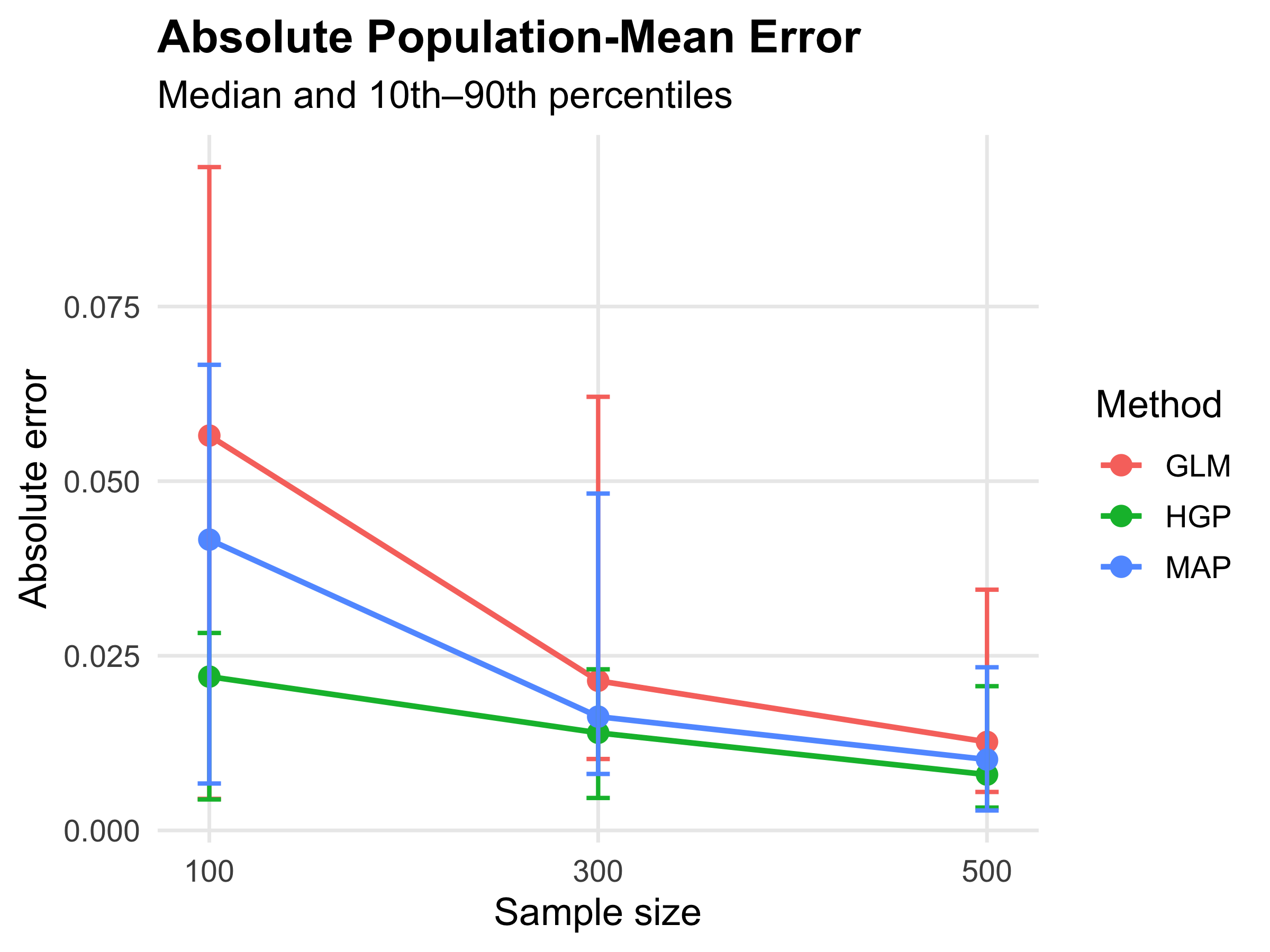}
\caption{Absolute population-mean prediction error under the
nonparametric GP data-generating mechanism as sample size varies.
Points show medians across simulation replicates, and vertical bars
indicate the 10th--90th percentiles. Population means are computed
using the true population counts $N_k$, with
$K_{\mathrm{frame}}=1800$.}
\label{fig:popmean-error}
\end{figure}

As shown in Figure~\ref{fig:popmean-error}, the median absolute population-mean prediction error decreases with sample size for all three methods. Using the true population counts $N_k$ as the benchmark, HGP achieves the lowest median error at each sample size considered ($n=100$, $300$, and $500$). This pattern highlights the advantage of HGP for recovering the population mean under the nonparametric GP data-generating mechanism, where its flexible representation of population structure translates into more accurate population-weighted predictions. The percentile intervals also indicate that aggregate prediction error varies across simulation replicates, particularly at the smaller sample sizes.

Taken together, the two diagnostics highlight the distinction between
weighted-mean agreement and population-mean prediction accuracy.
The former evaluates the mean implied by the estimated weights and
predicted rates, whereas the latter evaluates the aggregation of
predicted rates under the known population distribution. Their
rankings need not coincide because they use different weighting
schemes and therefore measure different properties of the fitted
methods. Alongside the MSE summaries and the group- and cell-level
plots, these results provide a more complete assessment of the
proposed HGP framework. In particular, the lower median aggregate
prediction errors observed for HGP under the nonparametric GP data-generating mechanism support
its usefulness in this setting, while the mean-gap results show that
improved prediction does not automatically imply closer agreement
for every weight-induced aggregate.

\subsection{Application to FFCWS Data}

We apply the proposed hierarchical Gaussian process (HGP) weighting framework to data
from the Future of Families and Child Wellbeing Study (FFCWS), a longitudinal birth cohort
study of children born in large U.S.\ cities between 1998 and 2000 (Reichman et al., \citeyear{ReichmanTeitlerGarfinkelMcLanahan2001}). A detailed description
of the survey design, sampling scheme, and construction of poststratification cells is
provided in our earlier work (Lee, \citeyear{Lee2026}). Here we briefly summarize the key features
relevant for the present analysis.

The FFCWS design induces a high-dimensional poststratification structure formed by cross-classifying maternal marital status, race/ethnicity, age, education, and city of birth (Carlson, \citeyear{Carlson2008}). This cross-classification yields \(K=3{,}840\) possible poststratification cells. Population information is available for \(3{,}339\) of these cells. We restrict attention to the \(n=3{,}424\) sampled observations whose corresponding cells have known population counts, covering \(K_{\mathrm{obs}}=1{,}059\) distinct poststratification cells within the population-supported subset. Thus, many population-supported cells are weakly observed or entirely unobserved in the sample.

Consistent with the simulation study, we incorporate reliable external benchmarks at a coarser level of aggregation rather than at the full \(K\)-cell resolution. Specifically, population-level information is incorporated for \(H=303\) hierarchical groups defined by race/ethnicity, age, and city, rather than separately for each of the \(1{,}059\) observed cells. This choice reflects common challenges in applied survey settings, where outcome information may be unavailable or unreliable at fine levels of stratification, necessitating the use of auxiliary information from external sources at aggregated resolutions. The hierarchical GP framework accommodates this structure by propagating information from coarser groups to finer cells through structured partial pooling.

For the outcome variable, we use a binary indicator of father interview participation in the FFCWS study. This indicator records whether a father was interviewed and is observed for every record in the analytic sample, including cases in which the father was not interviewed. We use this response indicator to model variation in survey participation across poststratification cells and to construct model-based survey weights; it is not treated as a substantive outcome of interest. The resulting weights are intended to account for differences in interview participation across observed covariate groups. The poststratification covariates consist exclusively of ordinal and nominal
categorical variables: ordinal variables are treated as numerical inputs, while nominal
variables are encoded using one-hot binary indicators (Garrido-Merch\'an and Hern\'andez-Lobato, \citeyear{GarridoMerchanHernandezLobato2020}). This representation facilitates construction of the squared exponential kernel used in the Gaussian process component. Together, the high dimensionality of the poststratification space, uneven
sample support across cells, and partial population coverage make FFCWS a demanding
environment for survey weighting. These features render it a natural empirical test case
for hierarchical Gaussian process smoothing, which is designed to borrow strength across
related cells while maintaining local adaptivity.

Figure~\ref{fig:final_plot_H} compares population benchmark rates
with rates obtained under two weighting schemes in the FFCWS
application. Both rows display the same $H$-groups, ordered by sample
size. The top row compares the population benchmark rates with
predicted rates aggregated using population cell counts $N_k$.
The bottom row retains the same population benchmark but uses
method-specific estimated weights $w_k$ to obtain the corresponding
weighted rates. This comparison allows us to examine whether the
local subgroup patterns change when population-count weighting is
replaced by estimated-weight weighting.

```latex
In the top row of Figure~\ref{fig:final_plot_H}, all three methods
capture aspects of the broad pattern of population benchmark rates
across the displayed $H$-groups, although discrepancies remain.
HGP and MAP Logit reproduce several local fluctuations, while
logistic regression shows more noticeable deviations in some groups.
Overall, the population-count-weighted predictions capture the
general pattern but do not consistently reproduce the benchmark
rates across groups.

\begin{figure}[H]
\centering
\includegraphics[width=0.98\textwidth]{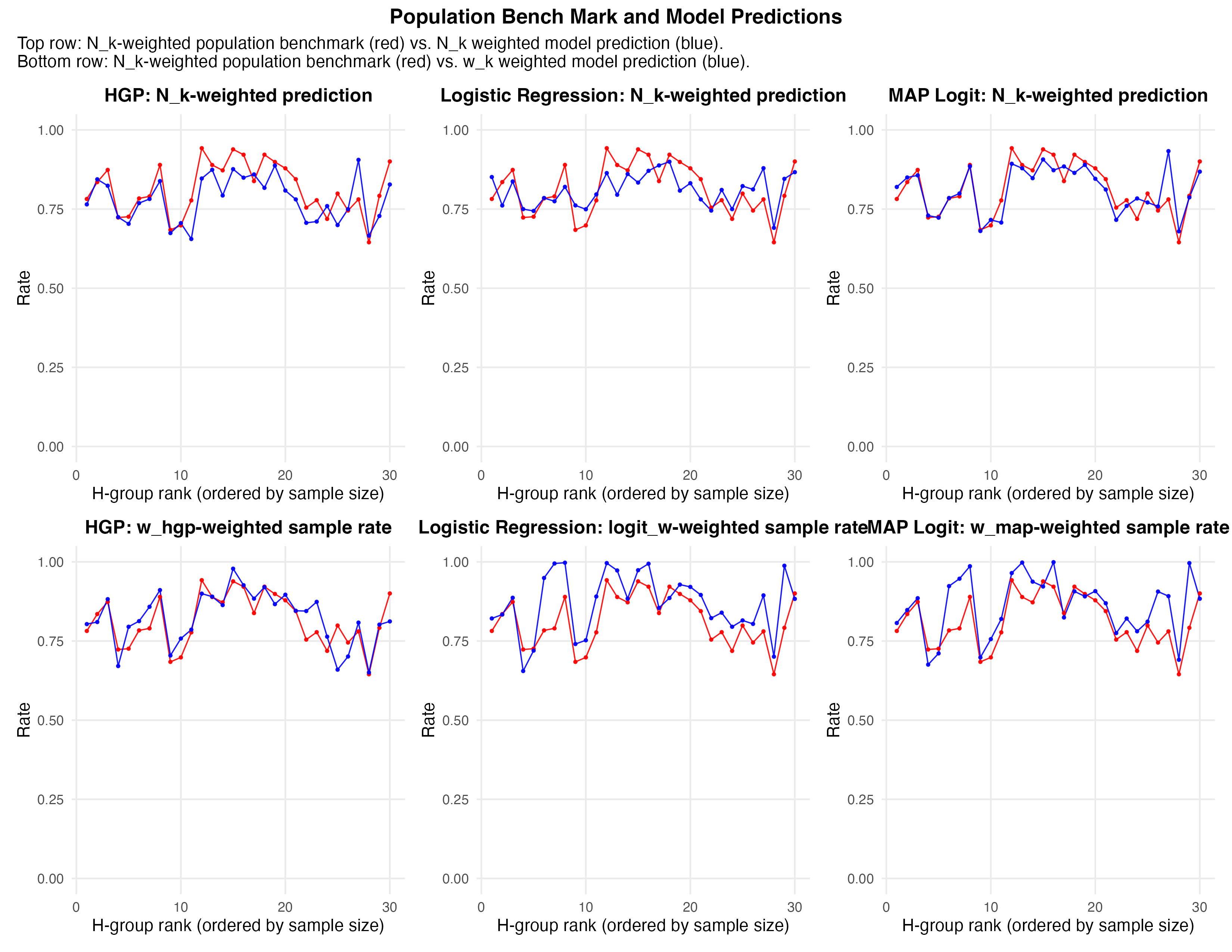}
\caption{Population benchmark and model-implied rates for the
top 30 $H$-groups ranked by sample size in the FFCWS application.
The top row compares the population benchmark rates (red) with
predicted rates (blue) aggregated using population cell counts
$N_k$. The bottom row retains the same population benchmark
(red) and compares it with rates obtained using method-specific
estimated weights $w_k$ (blue). Results are shown for the
hierarchical Gaussian process (HGP), logistic regression without
priors, and MAP Logit.}
\label{fig:final_plot_H}
\end{figure}

The bottom row examines the corresponding rates under estimated-weight
weighting. HGP shows the closest overall correspondence to the
population benchmark among the three methods, capturing many of the
local fluctuations with relatively small discrepancies in most
displayed groups. In contrast, logistic regression exhibits more
pronounced local departures, including several groups where the
estimated-weight-weighted rates approach or reach one. MAP Logit
also shows noticeable discrepancies in some groups, despite
capturing aspects of the broad pattern. These comparisons suggest
that HGP provides a closer representation of local subgroup rates
under estimated-weight weighting in this application, although
the degree of agreement varies across groups and methods.

The comparison between the two rows is informative because the red
population benchmark rates are unchanged. Differences between the
blue curves therefore reflect the change in the weighting scheme,
rather than a change in the benchmark or the displayed $H$-groups.
The results suggest that estimated-weight weighting can preserve
the broad subgroup-level pattern while producing local differences
in some groups. However, these visual comparisons alone do not
establish whether one weighting scheme yields more accurate
population estimates across the full poststratification space.

We next compare global mean agreement in the FFCWS application.
Table~\ref{tab:ffcws-mean-gap} reports the benchmark population mean
$\bar p^{\mathrm{pop}}$ and the model-implied weighted mean
$\bar p^w$ under each method. The relative gap is reported in
absolute value.

\begin{table}[H]
\centering
\caption{Population mean agreement in the FFCWS application.
The table compares the benchmark population mean
$\bar p^{\mathrm{pop}}$ with the model-implied weighted mean
$\bar p^w$. Relative gaps are reported in absolute value.}
\label{tab:ffcws-mean-gap}
\begin{tabular}{lccc}
\toprule
Method & $\bar p^{\mathrm{pop}}$ & $\bar p^w$ &
Absolute relative gap (\%) \\
\midrule
HGP              & 0.793 & 0.811 & 2.2 \\
GLM (no prior)   & 0.793 & 0.832 & 4.9 \\
Logit (MAP)      & 0.793 & 0.811 & 2.2 \\
\bottomrule
\end{tabular}
\end{table}

Table~\ref{tab:ffcws-mean-gap} shows that all three methods yield
model-implied weighted means relatively close to the benchmark
population mean. HGP and MAP Logit have the smallest absolute
relative gaps, both at 2.2\%, compared with 4.9\% for logistic
regression without priors. Thus, HGP and MAP Logit exhibit closer
agreement with the benchmark for this particular global functional.
This finding complements the subgroup-level diagnostics, which
reveal how the choice of weighting scheme affects local rates.
Nevertheless, agreement in the global mean does not necessarily
imply equally accurate rates within individual $H$-groups or
across the full poststratification space.

Taken together, the FFCWS application illustrates the importance
of evaluating weighting methods at both local and aggregate levels.
The subgroup-level comparisons reveal how predicted or weighted
rates vary across $H$-groups and how these patterns respond to
alternative weighting schemes, while the mean-gap analysis
summarizes agreement for a single population-level quantity.
These diagnostics provide complementary evidence about the behavior
of the proposed HGP-based weighting procedure in an empirical
survey setting.

\section{Discussion}

This paper develops a model-based survey-weighting framework using hierarchical Gaussian process (HGP) regression for structured populations with sparse poststratification cells. By combining hierarchical group effects with a kernel-based component, the framework accommodates nonlinear variation while borrowing information across related cells. The structured covariance representation also reduces the computational cost of posterior inference when the effective dimension of the model is substantially smaller than the sample size.

A central contribution is the extension of local equivalent-weight construction to the HGP setting. Rather than representing a nonlinear poststratified estimator as an exact weighted sample mean, we derive weights by matching its first-order sensitivity to that of a weighted sample-level quantity around the fitted latent-scale vector. The resulting weights provide a local, derivative-based representation of the model-based population estimand. Our theoretical results establish uniform convergence of the posterior mean under a stated posterior-tail condition and characterize the local equivalent weights under additional regularity assumptions. These results provide a foundation for the construction, but do not by themselves establish posterior contraction for the HGP model or guarantee stability of the weights without the accompanying assumptions.

The simulation results show that the relative performance of HGP and logistic regression depends on the underlying data-generating mechanism. Under correct parametric specification, unregularized logistic regression achieves the lowest population-weighted prediction error. Under the nonparametric Gaussian process mechanism, HGP achieves lower population-weighted prediction error than both logistic regression alternatives. Its advantage persists across the sample sizes and eligible-frame sizes examined. The subgroup- and cell-level diagnostics further illustrate the value of evaluating predictions at multiple resolutions, since agreement at a coarser group level may conceal discrepancies among individual poststratification cells.

The results also distinguish predictive accuracy from agreement in a global population mean. Although HGP performs well under the nonparametric simulation mechanism, it does not uniformly yield the smallest gap between the model-implied weighted mean and the population benchmark. The FFCWS application similarly shows that all three methods capture broad patterns in the population benchmarks, with differences across groups and cells; HGP and MAP Logit have the smallest reported relative gaps for the global mean. These findings underscore that no single diagnostic fully characterizes a weighting method: cell-level prediction, population-weighted error, and agreement for a particular aggregate estimand address related but distinct objectives.

Several limitations motivate further work. The theoretical results rely on stated convergence and regularity conditions, and the simulations cover a limited set of population structures and sampling settings. The empirical analysis also uses available population benchmarks and a particular participation indicator, so its findings should not be interpreted as evidence that HGP automatically eliminates nonresponse bias or resolves incomplete population coverage. Future work could investigate broader sampling and missingness mechanisms, the sensitivity of the resulting weights to model and kernel choices, and the use of these weights for downstream estimation and uncertainty quantification.

Overall, HGP provides a flexible extension of local equivalent-weight construction for structured populations. Its principal advantage in the settings studied is improved population-weighted prediction when the underlying structure departs from a parametric logistic specification. The results also highlight the importance of assessing predictive performance and aggregate estimation separately when evaluating model-based survey weights.

\appendix
\section{Proofs}
\label{app:proofs}

\singlespacing

\subsection{Proof of Theorem 1}

\begin{proof}
\setlength{\abovedisplayskip}{4pt plus 1pt minus 1pt}
\setlength{\belowdisplayskip}{4pt plus 1pt minus 1pt}
\setlength{\abovedisplayshortskip}{2pt plus 1pt minus 1pt}
\setlength{\belowdisplayshortskip}{3pt plus 1pt minus 1pt}
\setlength{\jot}{2pt}

By the hierarchical Gaussian process prior, the latent function values at the observed and prediction locations are jointly Gaussian. In particular, for the latent function values $\mathbf f$ and $\mathbf f_*$,
\[
\begin{bmatrix}
    \mathbf f \\
    \mathbf f_*
\end{bmatrix}
\sim
\mathcal{N}\left(
\begin{bmatrix}
    \boldsymbol{\mu}_{\text{prior},h} \\
    \boldsymbol{\mu}_{\text{prior},h_*}
\end{bmatrix},
\begin{bmatrix}
    \boldsymbol{\Sigma}_{\mu}[h,h']
        + \mathbf K_{\eta}[k,k']
    &
    \boldsymbol{\Sigma}_{\mu}[h,h_*]
        + \mathbf K_{\eta}[k,k_*]
    \\
    \boldsymbol{\Sigma}_{\mu}[h_*,h]
        + \mathbf K_{\eta}[k_*,k]
    &
    \boldsymbol{\Sigma}_{\mu}[h_*,h_*']
        + \mathbf K_{\eta}[k_*,k_*']
\end{bmatrix}
\right).
\]
For convenience, denote this joint distribution by
\[
\begin{bmatrix}
    \mathbf f \\
    \mathbf f_*
\end{bmatrix}
\sim
\mathcal{N}\left(
\begin{bmatrix}
    \boldsymbol{\mu}_{f} \\
    \boldsymbol{\mu}_{f_*}
\end{bmatrix},
\begin{bmatrix}
    \boldsymbol{\Sigma}_{ff}
    & \boldsymbol{\Sigma}_{f f_*} \\
    \boldsymbol{\Sigma}_{f_* f}
    & \boldsymbol{\Sigma}_{f_* f_*}
\end{bmatrix}
\right).
\]
Applying the standard conditional Gaussian formula gives
\[
f_* \mid \mathbf f
\sim
\mathcal{N}\left(
    \boldsymbol{\mu}_{f_*}
    + \boldsymbol{\Sigma}_{f_* f}
      \boldsymbol{\Sigma}_{ff}^{-1}
      (\mathbf f - \boldsymbol{\mu}_{f}),
    \;
    \boldsymbol{\Sigma}_{f_* f_*}
    - \boldsymbol{\Sigma}_{f_* f}
      \boldsymbol{\Sigma}_{ff}^{-1}
      \boldsymbol{\Sigma}_{f f_*}
\right).
\]
This establishes the conditional Gaussian distribution of the latent function at the prediction locations given its values at the observed locations.

For continuous outcomes under the identity link, the observation model is
\[
    \mathbf y = \mathbf f + \boldsymbol{\varepsilon},
    \qquad
    \boldsymbol{\varepsilon}
    \sim \mathcal{N}(\mathbf 0, \sigma^2 \mathbf I_n),
\]
where the observation noise is independent of the latent function. Consequently, the joint distribution of the observed responses and latent predictions is
\[
\begin{bmatrix}
    \mathbf y \\
    \mathbf f_*
\end{bmatrix}
\sim
\mathcal{N}\left(
\begin{bmatrix}
    \boldsymbol{\mu}_{f} \\
    \boldsymbol{\mu}_{f_*}
\end{bmatrix},
\begin{bmatrix}
    \boldsymbol{\Sigma}_{ff} + \sigma^2\mathbf I_n
    & \boldsymbol{\Sigma}_{f f_*} \\
    \boldsymbol{\Sigma}_{f_* f}
    & \boldsymbol{\Sigma}_{f_* f_*}
\end{bmatrix}
\right).
\]
Defining
\[
    \boldsymbol{\Sigma}_{yy}
    = \boldsymbol{\Sigma}_{ff} + \sigma^2\mathbf I_n,
\]
the cross-covariance matrices satisfy
\[
    \operatorname{Cov}(\mathbf f_*,\mathbf y)
    = \boldsymbol{\Sigma}_{f_* f},
    \qquad
    \operatorname{Cov}(\mathbf y,\mathbf f_*)
    = \boldsymbol{\Sigma}_{f f_*}.
\]
Applying the conditional Gaussian formula again yields
\[
f_* \mid \mathbf y
\sim
\mathcal{N}\left(
    \boldsymbol{\mu}_{f_*}
    + \boldsymbol{\Sigma}_{f_* f}
      \boldsymbol{\Sigma}_{yy}^{-1}
      (\mathbf y - \boldsymbol{\mu}_{f}),
    \;
    \boldsymbol{\Sigma}_{f_* f_*}
    - \boldsymbol{\Sigma}_{f_* f}
      \boldsymbol{\Sigma}_{yy}^{-1}
      \boldsymbol{\Sigma}_{f f_*}
\right).
\]
Finally, the future continuous observation satisfies
\[
    \mathbf y_* = \mathbf f_* + \boldsymbol{\varepsilon}_*,
    \qquad
    \boldsymbol{\varepsilon}_*
    \sim \mathcal{N}(\mathbf 0, \sigma^2\mathbf I_{n_*}),
\]
where $\boldsymbol{\varepsilon}_*$ is independent of the training observations and latent function. Therefore, the predictive distribution of a future observation is
\[
\mathbf y_* \mid \mathbf y
\sim
\mathcal{N}\left(
    \mathbb E[\mathbf f_* \mid \mathbf y],
    \operatorname{Cov}(\mathbf f_* \mid \mathbf y)
    + \sigma^2\mathbf I_{n_*}
\right).
\]
For binary outcomes under a non-identity link, the latent function retains the Gaussian prior and the covariance structure derived above. However, the observed responses follow the Bernoulli likelihood
\[
    Y_i \mid f_i
    \sim \operatorname{Bernoulli}\bigl(g^{-1}(f_i)\bigr).
\]
Consequently, conditioning on observed binary responses generally does not yield a Gaussian posterior. Thus, the Gaussian conditioning result applies directly to the latent function values under the prior, whereas posterior prediction given observed binary responses must additionally account for the Bernoulli likelihood.
\end{proof}

\subsection{Proof of Theorem 2}

\begin{proof}
\setlength{\abovedisplayskip}{4pt plus 1pt minus 1pt}
\setlength{\belowdisplayskip}{4pt plus 1pt minus 1pt}
\setlength{\abovedisplayshortskip}{2pt plus 1pt minus 1pt}
\setlength{\belowdisplayshortskip}{3pt plus 1pt minus 1pt}
\setlength{\jot}{2pt}

Write the model as
\[
\mathbf f+\boldsymbol{\epsilon}
=
Z_\mu\boldsymbol{\mu}_h
+
Z_\eta\boldsymbol{\eta}_{k,\mathrm{obs}}
+
\boldsymbol{\epsilon},
\]
where $Z_\mu\in\mathbb R^{n\times H}$ and
$Z_\eta\in\mathbb R^{n\times K_{\mathrm{obs}}}$ are group incidence matrices with entries
\[
(Z_\mu)_{ih}=\mathbf 1\{h[i]=h\},
\qquad
(Z_\eta)_{ik}=\mathbf 1\{k[i]=k\},
\]
and $\boldsymbol{\epsilon}\sim
\mathcal N(\mathbf 0,\sigma^2 I_n)$.
Here, $K_{\mathrm{obs}}$ denotes the number of distinct kernel groups represented in the sample. The covariance matrix of the observed responses is therefore
\[
V
=
Z_\mu\Sigma_\mu Z_\mu^\top
+
Z_\eta K_{\eta,\mathrm{obs}}Z_\eta^\top
+
\sigma^2 I_n,
\]
where $K_{\eta,\mathrm{obs}}$ is the kernel covariance matrix restricted to the observed kernel groups.

Define
\[
U=[Z_\mu\;\;Z_\eta]
\in\mathbb R^{n\times r},
\qquad
C=\operatorname{diag}(\Sigma_\mu,K_{\eta,\mathrm{obs}}),
\qquad
r=H+K_{\mathrm{obs}}.
\]
Then
\[
V=\sigma^2 I_n+UCU^\top.
\]
Using the Woodbury identity, and assuming $\sigma^2>0$ and $C$ is invertible, we obtain
\[
\begin{aligned}
V^{-1}
&=(\sigma^2 I_n+UCU^\top)^{-1}\\
&=\sigma^{-2}I_n
-\sigma^{-2}U
\left(C^{-1}+\sigma^{-2}U^\top U\right)^{-1}
U^\top\sigma^{-2}.
\end{aligned}
\]

For dense $U$, forming $U^\top U$ costs $\mathcal O(nr^2)$, while factorizing the resulting $r\times r$ matrix $C^{-1}+\sigma^{-2}U^\top U$ using Cholesky decomposition costs $\mathcal O(r^3)$. Thus, the cost of forming and factorizing the reduced system is
\[
\begin{aligned}
\mathcal O(nr^2+r^3)
&=\mathcal O\left(r^2(n+r)\right)\\
&=\mathcal O\left(
(H+K_{\mathrm{obs}})^2
(n+H+K_{\mathrm{obs}})
\right).
\end{aligned}
\]
This bound excludes the additional computation required to construct or factorize $C$, as well as the cost of prediction across the full poststratification space. When $H+K_{\mathrm{obs}}\ll n$, the reduced system can be substantially less expensive than a direct dense factorization of $V$, which costs $\mathcal O(n^3)$. The sparsity of the incidence matrices may further reduce the computational cost.
\end{proof}

\subsection{Proof of Proposition 1}

\begin{proof}
\setlength{\abovedisplayskip}{4pt plus 1pt minus 1pt}
\setlength{\belowdisplayskip}{4pt plus 1pt minus 1pt}
\setlength{\abovedisplayshortskip}{2pt plus 1pt minus 1pt}
\setlength{\belowdisplayshortskip}{3pt plus 1pt minus 1pt}
\setlength{\jot}{2pt}

Under $\Sigma_\mu\to0$ and
$\boldsymbol{\mu}_{\mathrm{prior}}=\mathbf{0}$, the latent function reduces to $f=\eta\sim\mathcal{GP}(0,k_\eta)$. With the identity link and Gaussian noise, the observed responses satisfy
\[
\mathbf{y}\sim
\mathcal{N}(\mathbf{0},\mathbf{K}_X+\sigma^2 I_n),
\]
where $\mathbf{K}_X\in\mathbb{R}^{n\times n}$ has entries
$(\mathbf{K}_X)_{ij}=k_\eta(x_i,x_j)$.
The standard GP posterior mean at $x^*$ is therefore
\[
\mathbb{E}[f(x^*)\mid\mathbf{y}]
=
\mathbf{k}_\eta(x^*,X)^\top
(\mathbf{K}_X+\sigma^2 I_n)^{-1}\mathbf{y},
\]
where $\mathbf{k}_\eta(x^*,X)\in\mathbb{R}^n$ has entries
$\mathbf{k}_\eta(x^*,X)_i=k_\eta(x^*,x_i)$.

Let $\mathcal{H}_{k_\eta}$ denote the reproducing kernel Hilbert space (RKHS) associated with $k_\eta$. Kernel ridge regression (KRR) solves
\[
\hat f
=
\arg\min_{f\in\mathcal{H}_{k_\eta}}
\left\{
\sum_{i=1}^n (y_i-f(x_i))^2
+\lambda\|f\|_{\mathcal{H}_{k_\eta}}^2
\right\}
\]
and admits the dual representation
\[
\hat f(x^*)
=
\mathbf{k}_\eta(x^*,X)^\top
(\mathbf{K}_X+\lambda I_n)^{-1}\mathbf{y}.
\]
Setting $\lambda=\sigma^2$ makes the GP posterior mean and the KRR prediction identical, establishing the proposition.
\end{proof}

\subsection{Proof of Proposition 2}

\begin{proof}
\setlength{\abovedisplayskip}{4pt plus 1pt minus 1pt}
\setlength{\belowdisplayskip}{4pt plus 1pt minus 1pt}
\setlength{\abovedisplayshortskip}{2pt plus 1pt minus 1pt}
\setlength{\belowdisplayshortskip}{3pt plus 1pt minus 1pt}
\setlength{\jot}{2pt}

By construction,
\[
\boldsymbol{\eta}\sim
\mathcal{N}(\mathbf{0},K_\eta)
=
\mathcal{N}(\mathbf{0},\tau^2XX^\top).
\]
Introducing
\[
\mathbf{b}_1\sim\mathcal{N}(\mathbf{0},\tau^2I_p)
\]
and setting $\boldsymbol{\eta}=X\mathbf{b}_1$ yields
\[
\begin{aligned}
\operatorname{Cov}(\boldsymbol{\eta})
&=X\operatorname{Cov}(\mathbf{b}_1)X^\top\\
&=\tau^2XX^\top=K_\eta,
\qquad
\mathbb{E}[\boldsymbol{\eta}]=\mathbf{0}.
\end{aligned}
\]
Thus, the GP component has the same distribution as the random-slope component. Retaining the same group-level intercepts and observation noise gives the stated representation.
\end{proof}

\subsection{Proof of Lemma 1}

\begin{proof}
\setlength{\abovedisplayskip}{4pt plus 1pt minus 1pt}
\setlength{\belowdisplayskip}{4pt plus 1pt minus 1pt}
\setlength{\abovedisplayshortskip}{2pt plus 1pt minus 1pt}
\setlength{\belowdisplayshortskip}{3pt plus 1pt minus 1pt}
\setlength{\jot}{2pt}

By Jensen's inequality,
\[
\|\hat f_n-f_0\|_\infty
\le \mathbb E_{\Pi_n}\|f-f_0\|_\infty.
\]
Decomposing the posterior expectation over $A_n$ and $A_n^c$, we obtain
\[
\begin{aligned}
\mathbb E_{\Pi_n}\|f-f_0\|_\infty
&=
\mathbb E_{\Pi_n}\!\left[
\|f-f_0\|_\infty\mathbf 1_{A_n}
\right]\\
&\quad+
\mathbb E_{\Pi_n}\!\left[
\|f-f_0\|_\infty\mathbf 1_{A_n^c}
\right]\\
&\le M\varepsilon_n+O_{P_{f_0}}(\varepsilon_n)\\
&=O_{P_{f_0}}(\varepsilon_n).
\end{aligned}
\]
\end{proof}

\subsection{Proof of Corollary 1}

\begin{proof}
\setlength{\abovedisplayskip}{4pt plus 1pt minus 1pt}
\setlength{\belowdisplayskip}{4pt plus 1pt minus 1pt}
\setlength{\abovedisplayshortskip}{2pt plus 1pt minus 1pt}
\setlength{\belowdisplayshortskip}{3pt plus 1pt minus 1pt}
\setlength{\jot}{2pt}

Let $\phi=g^{-1}$. By the definitions of $\hat\mu_n$ and $\mu_0$, for every $x\in\mathcal X$,
\[
\hat\mu_n(x)-\mu_0(x)
=
\mathbb E_{\Pi_n}[\phi(f(x))]-\phi(f_0(x)).
\]
Applying Jensen's inequality gives
\[
\begin{aligned}
|\hat\mu_n(x)-\mu_0(x)|
&\le
\mathbb E_{\Pi_n}
\left[|\phi(f(x))-\phi(f_0(x))|\right].
\end{aligned}
\]
Define $Z_x:=|\phi(f(x))-\phi(f_0(x))|$. By nonnegativity,
\[
\sup_{x\in\mathcal X}\mathbb E_{\Pi_n}[Z_x]
\le
\mathbb E_{\Pi_n}\left[\sup_{x\in\mathcal X}Z_x\right].
\]
Taking the supremum over $x\in\mathcal X$ therefore yields
\[
\begin{aligned}
\|\hat\mu_n-\mu_0\|_\infty
&\le
\mathbb E_{\Pi_n}\left[
\sup_{x\in\mathcal X}
|\phi(f(x))-\phi(f_0(x))|
\right].
\end{aligned}
\]
By the mean value theorem, for every $x\in\mathcal X$, there exists a point $\xi_x$ between $f(x)$ and $f_0(x)$ such that
\[
|\phi(f(x))-\phi(f_0(x))|
=
|\phi'(\xi_x)|\,|f(x)-f_0(x)|.
\]
The assumed uniform derivative bound implies
\[
|\phi(f(x))-\phi(f_0(x))|
\le L|f(x)-f_0(x)|.
\]
Taking the supremum over $x$ gives
\[
\sup_{x\in\mathcal X}
|\phi(f(x))-\phi(f_0(x))|
\le L\|f-f_0\|_\infty.
\]
Consequently,
\[
\|\hat\mu_n-\mu_0\|_\infty
\le L\,\mathbb E_{\Pi_n}\|f-f_0\|_\infty.
\]
By Lemma~\ref{lem:posterior_mean_rate},
\[
\mathbb E_{\Pi_n}\|f-f_0\|_\infty
=O_{P_{f_0}}(\varepsilon_n).
\]
Since $L$ is a fixed finite constant, multiplication by $L$ preserves this stochastic order. Hence,
\[
\boxed{
\|\hat\mu_n-\mu_0\|_\infty
=O_{P_{f_0}}(\varepsilon_n).
}
\]
\end{proof}

\subsection{Proof of Lemma 2}

\begin{proof}
\setlength{\abovedisplayskip}{4pt plus 1pt minus 1pt}
\setlength{\belowdisplayskip}{4pt plus 1pt minus 1pt}
\setlength{\abovedisplayshortskip}{2pt plus 1pt minus 1pt}
\setlength{\belowdisplayshortskip}{3pt plus 1pt minus 1pt}
\setlength{\jot}{2pt}

By the uniform consistency of $\hat f_n$,
\[
\begin{aligned}
\|z-\hat z\|_\infty
&=\max_{1\le i\le n}
|f_0(x_i)-\hat f_n(x_i)|\\
&\leq \|\hat f_n-f_0\|_\infty\\
&=O_{P_{f_0}}(\varepsilon_n)
=o_{P_{f_0}}(1).
\end{aligned}
\]
Since $I_0$ is a compact subset of $I^\circ$, there exists a compact set $J$ such that
\[
I_0\subset J\subset I^\circ.
\]
Uniform consistency then implies
\[
P_{f_0}\big(\hat z_i\in J
\text{ for all }i=1,\ldots,n\big)\to 1.
\]
On this event, both $z_i$ and $\hat z_i$ belong to $J$, as does every point between them.

For each $i$, the mean value theorem gives a point $\xi_i$ between $z_i$ and $\hat z_i$ such that
\[
\phi(z_i)-\phi(\hat z_i)
=\phi'(\xi_i)(z_i-\hat z_i).
\]
Define
\[
r_{g,i}:=
\big(\phi'(\xi_i)-\phi'(\hat z_i)\big)
(z_i-\hat z_i),
\qquad
r_g:=(r_{g,1},\ldots,r_{g,n})^\top.
\]
It follows that
\[
\phi(z)-\phi(\hat z)
=A(\hat z)(z-\hat z)+r_g.
\]

Because $\phi'$ is continuous on $I$, it is uniformly continuous on the compact set $J$. Its modulus of continuity,
\[
\omega_J(a):=
\sup_{\substack{u,v\in J\\|u-v|\le a}}
|\phi'(u)-\phi'(v)|,
\]
therefore satisfies $\omega_J(a)\to 0$ as $a\downarrow 0$. On the event that all $z_i$ and $\hat z_i$ belong to $J$, we have
\[
|\xi_i-\hat z_i|
\le |z_i-\hat z_i|
\le \|z-\hat z\|_\infty.
\]
Consequently,
\[
\begin{aligned}
\|r_g\|_\infty
&\le
\|z-\hat z\|_\infty
\max_{1\le i\le n}
|\phi'(\xi_i)-\phi'(\hat z_i)|\\
&\le
\|z-\hat z\|_\infty
\omega_J(\|z-\hat z\|_\infty).
\end{aligned}
\]
Since $\|z-\hat z\|_\infty=o_{P_{f_0}}(1)$ and $\omega_J(a)\to0$ as $a\downarrow0$, it follows that
\[
\|r_g\|_\infty
=o_{P_{f_0}}\big(\|z-\hat z\|_\infty\big).
\]
This bound also holds when $\|z-\hat z\|_\infty=0$, in which case $r_g=0$. Finally, since $\|z-\hat z\|_\infty=o_{P_{f_0}}(1)$, we obtain
\[
\|r_g\|_\infty=o_{P_{f_0}}(1).
\]
\end{proof}

\subsection{Proof of Theorem 3}

\begin{proof}
\setlength{\abovedisplayskip}{4pt plus 1pt minus 1pt}
\setlength{\belowdisplayskip}{4pt plus 1pt minus 1pt}
\setlength{\abovedisplayshortskip}{2pt plus 1pt minus 1pt}
\setlength{\belowdisplayshortskip}{3pt plus 1pt minus 1pt}
\setlength{\jot}{2pt}

Write
\[
\phi\equiv g^{-1},
\qquad
B\equiv\Sigma_{f_*f}\Sigma_{yy}^{-1},
\qquad
h\equiv z-\hat z.
\]
Since $m_*(z)$ is affine in $z$,
\[
m_*(z)-m_*(\hat z)=Bh.
\]
At the cell level, we apply the inverse-link linearization principle of Lemma~\ref{lem:ginv_remainder} to the $K$-dimensional vectors $m_*(z)$ and $m_*(\hat z)$. A first-order Taylor expansion of $\phi$ componentwise at $m_*(\hat z)$ gives
\[
\phi(m_*(z))-\phi(m_*(\hat z))
=
D_*Bh+r_{g,*},
\]
where the fixed-dimensional Taylor remainder satisfies
\[
\|r_{g,*}\|_\infty
=o_{P_{f_0}}(\|h\|_\infty).
\]
Indeed, the assumed boundedness of $B$ implies
\[
\|Bh\|_\infty
\leq \|B\|_{\infty}\|h\|_\infty
=O(\|h\|_\infty),
\]
while the uniform continuity of $\phi'$ on the specified neighborhood yields the stated remainder bound.

Substituting this expansion into the definition of $T$ gives
\[
\begin{aligned}
T(z)-T(\hat z)
&=
\frac{1}{N}(N^{\mathrm{pop}})^\top D_*Bh
+
\frac{1}{N}(N^{\mathrm{pop}})^\top r_{g,*}.
\end{aligned}
\]
Because the population cell counts are nonnegative and sum to $N$,
\[
\begin{aligned}
\left|
\frac{1}{N}(N^{\mathrm{pop}})^\top r_{g,*}
\right|
&\leq \|r_{g,*}\|_\infty\\
&=o_{P_{f_0}}(\|h\|_\infty).
\end{aligned}
\]
Consequently,
\[
T(z)-T(\hat z)
=
\nabla_z T(\hat z)^\top h
+
o_{P_{f_0}}(\|h\|_\infty),
\]
where the chain rule gives
\[
\begin{aligned}
\nabla_z T(\hat z)
&=
\frac{1}{N}B^\top D_*N^{\mathrm{pop}}\\
&=
\frac{1}{N}
\Sigma_{yy}^{-\top}
\Sigma_{ff_*}D_*N^{\mathrm{pop}}.
\end{aligned}
\tag{1}
\]

We next identify the local equivalent weights by matching the first-order directional derivatives of $T$ and the weighted comparison quantity. For an arbitrary direction $h\in\mathbb R^n$, consider the path
\[
z(t)=\hat z+th.
\]
By the preceding expansion and the chain rule,
\[
\left.\frac{d}{dt}T(z(t))\right|_{t=0}
=
\nabla_z T(\hat z)^\top h.
\]
For
\[
S_w(z)=\frac{1}{n}w^\top\phi(z),
\]
the chain rule similarly yields
\[
\left.\frac{d}{dt}S_w(z(t))\right|_{t=0}
=
\frac{1}{n}w^\top A(\hat z)h.
\]
By the definition of local equivalent weights, these derivatives must agree for every direction $h$. Thus,
\[
\nabla_z T(\hat z)^\top h
=
\frac{1}{n}w^\top A(\hat z)h
\qquad\text{for every }h\in\mathbb R^n.
\]
Equating the coefficient vectors gives
\[
\nabla_z T(\hat z)
=
\frac{1}{n}A(\hat z)^\top w.
\]
Since $A(\hat z)$ is diagonal and nonsingular,
\[
w=nA(\hat z)^{-1}\nabla_z T(\hat z).
\tag{2}
\]
This establishes the weight formula and the uniqueness of the first-order coefficient vector.

It remains to verify the local survey-weight identity. By Lemma~\ref{lem:ginv_remainder},
\[
\phi(z)-\phi(\hat z)
=
A(\hat z)h+r_g,
\qquad
\|r_g\|_\infty
=o_{P_{f_0}}(\|h\|_\infty).
\]
Since $\|h\|_\infty=o_{P_{f_0}}(1)$, we also have
\[
\|r_g\|_\infty=o_{P_{f_0}}(1).
\]
Multiplying the expansion by $w^\top/n$ yields
\[
\begin{aligned}
\frac{1}{n}w^\top
\{\phi(z)-\phi(\hat z)\}
&=
\frac{1}{n}w^\top A(\hat z)h
+\frac{1}{n}w^\top r_g.
\end{aligned}
\]
By the assumed weight-norm condition, the weighted remainder satisfies
\[
\begin{aligned}
\left|\frac{1}{n}w^\top r_g\right|
&\leq
\frac{\|w\|_1}{n}\|r_g\|_\infty\\
&=
o_{P_{f_0}}(\|h\|_\infty).
\end{aligned}
\]
Moreover, equation (2) implies
\[
\frac{1}{n}w^\top A(\hat z)h
=
\nabla_z T(\hat z)^\top h.
\]
Combining this equality with the expansion in (1) gives
\[
\begin{aligned}
T(z)-T(\hat z)
&=
\frac{1}{n}w^\top
\{\phi(z)-\phi(\hat z)\}
+
o_{P_{f_0}}(\|z-\hat z\|_\infty),
\end{aligned}
\]
which proves \eqref{eq:caseB_identity}. Finally, substituting (1) into (2) gives the stated closed-form expression for the local equivalent unit weights.
\end{proof}

\doublespacing

\section{Data and Code Availability}

Selected R code used to reproduce the simulations and empirical analyses
in this paper is publicly available at \url{https://github.com/huneel777/logit_equivalent_weight/tree/main/hgp_weight_code}.

The empirical application uses data from the Future of Families and
Child Wellbeing Study (FFCWS). Public-use data and documentation are
available at
\href{https://ffcws.princeton.edu/documentation}
{https://ffcws.princeton.edu/documentation}.

\begingroup
\setlength{\bibsep}{3pt}

\endgroup

\end{document}